\documentclass[onecolumn,journal]{IEEEtran}
\IEEEoverridecommandlockouts
\usepackage[T1]{fontenc}
\usepackage{url}
\usepackage{multirow}
\usepackage{cite}
\usepackage{stfloats}
\usepackage{graphicx}
\usepackage{amssymb}
\usepackage{eucal}
\usepackage{color}
\usepackage{comment}
\usepackage{arydshln}
\usepackage{algorithm}
\usepackage{algpseudocode}
\usepackage{xcolor}
\algnewcommand\algorithmicinput{\textbf{Input:}}
\algnewcommand\INPUT{\item[\algorithmicinput]}
\algnewcommand\algorithmicoutput{\textbf{Output:}}
\algnewcommand\OUTPUT{\item[\algorithmicoutput]}
\usepackage{mathtools}

\usepackage{enumerate}

\usepackage{setspace}
\usepackage{amsthm}
\theoremstyle{definition}

\newtheorem{definition}{Definition}
\newtheorem{proposition}{Proposition}

\newtheorem{theorem}{Theorem}

\newtheorem{remark}{Remark}
\newtheorem{example}{Example}

\newcommand{\remove}[1]{}
\usepackage{tikz}
\usetikzlibrary{positioning}

\newcommand\nc\newcommand
\nc{\vzero}{{\boldsymbol{0}}}
\nc{\bfC}{{\boldsymbol C}}
\nc{\bfc}{{\boldsymbol c}}
\nc{\bfr}{{\boldsymbol r}}
\nc\bfa{{\boldsymbol a}}\nc\bfA{{\boldsymbol A}}\nc\cA{{\mathcal A}}
\nc\bfb{{\boldsymbol b}}\nc\bfB{{\boldsymbol B}}\nc\cB{{\mathcal B}}
\nc\bfd{{\boldsymbol d}}\nc\bfD{{\boldsymbol D}}\nc\cD{{\mathcal D}}
\nc\bfe{{\boldsymbol e}}\nc\bfE{{\boldsymbol E}}\nc\cE{{\mathcal E}}
\nc\bff{{\boldsymbol f}}\nc\bfF{{\boldsymbol F}}\nc\cF{{\mathcal F}}
\nc\bfg{{\boldsymbol g}}\nc\bfG{{\boldsymbol G}}\nc\cG{{\mathcal G}}
\nc\bfh{{\boldsymbol h}}\nc\bfH{{\boldsymbol H}}\nc\cH{{\mathcal H}}
\nc\bfi{{\boldsymbol i}}\nc\bfI{{\boldsymbol I}}\nc\cI{{\mathcal I}}
\nc\bfj{{\boldsymbol j}}\nc\bfJ{{\boldsymbol J}}\nc\cJ{{\mathcal J}}
\nc\bfk{{\boldsymbol k}}\nc\bfK{{\boldsymbol K}}\nc\cK{{\mathcal K}}
\nc\bfl{{\boldsymbol l}}\nc\bfL{{\boldsymbol L}}\nc\cL{{\mathcal L}}
\nc\bfm{{\boldsymbol m}}\nc\bfM{{\boldsymbol M}}\nc\cM{{\mathcal M}}
\nc\bfn{{\boldsymbol n}}\nc\bfN{{\boldsymbol N}}\nc\cN{{\mathcal N}}
\nc\bfo{{\boldsymbol o}}\nc\bfO{{\boldsymbol O}}\nc\cO{{\mathcal O}}
\nc\bfp{{\boldsymbol p}}\nc\bfP{{\boldsymbol P}}\nc\cP{{\mathcal P}}
\nc\bfq{{\boldsymbol q}}\nc\bfQ{{\boldsymbol Q}}\nc\cQ{{\mathcal Q}}
\nc\bfs{{\boldsymbol s}}\nc\bfS{{\boldsymbol S}}\nc\cS{{\mathcal S}}
\nc\bft{{\boldsymbol t}}\nc\bfT{{\boldsymbol T}}\nc\cT{{\mathcal T}}
\nc\bfu{{\boldsymbol u}}\nc\bfU{{\boldsymbol U}}\nc\cU{{\mathcal U}}
\nc\bfv{{\boldsymbol v}}\nc\bfV{{\boldsymbol V}}\nc\cV{{\mathcal V}}
\nc\bfw{{\boldsymbol w}}\nc\bfW{{\boldsymbol W}}\nc\cW{{\mathcal W}}
\nc\bfx{{\boldsymbol x}}\nc\bfX{{\boldsymbol X}}\nc\cX{{\mathcal X}}
\nc\bfy{{\boldsymbol y}}\nc\bfY{{\boldsymbol Y}}\nc\cY{{\mathcal Y}}
\nc\bfz{{\boldsymbol z}}\nc\bfZ{{\boldsymbol Z}}\nc\cZ{{\mathcal Z}}

\newcommand{\FF}{\mathbb{F}}

\newcommand{\ZZ}{\mathbb{Z}}

\nc{\Read}{{\sf R}}
\nc{\cost}{{\sf cost}}
\nc{\sqs}{{\rm SQS}}

\newcommand{\diff}{{\rm diff}}

\usepackage{hyperref}
\newcommand\myshade{70} 
\hypersetup{ 
	linkcolor  = red!\myshade!black,
	citecolor  = blue!\myshade!black,
	urlcolor   = blue!\myshade!black,
	colorlinks = true,
}
\title{Repairing with Zero Skip Cost} %\\[-2mm]}
\author{

  \IEEEauthorblockN{ Wenqin Zhang\IEEEauthorrefmark{1},
  Yeow Meng Chee\IEEEauthorrefmark{2},
  Son Hoang Dau\IEEEauthorrefmark{3}, 
  Tuvi Etzion\IEEEauthorrefmark{2}\IEEEauthorrefmark{4},
  Han Mao Kiah\IEEEauthorrefmark{5}, 
Yuan Luo\IEEEauthorrefmark{1} 
 }\\
   \IEEEauthorblockA{\small \IEEEauthorrefmark{1}School of Electronic Information and Electrical Engineering, Shanghai Jiao Tong University, Shanghai, China
    }\\
\IEEEauthorblockA{\small 
 \IEEEauthorrefmark{2}Dept. of Industrial Systems Engineering and Management, National University of Singapore
    }\\
\IEEEauthorblockA{\small 
 \IEEEauthorrefmark{3}School of Computing Technologies, RMIT University, Melbourne, VIC, Australia
    }\\
 \IEEEauthorblockA{\small \IEEEauthorrefmark{4}Computer Science Faculty, Technion, Israel Institute of Technology,
    } \\   
 \IEEEauthorblockA{\small 
 \IEEEauthorrefmark{5}School of Physical and Mathematical Sciences, Nanyang Technological University, Singapore
    }

  {\footnotesize wenqin\_zhang@sjtu.edu.cn,
  ymchee@nus.edu.sg, sonhoang.dau@rmit.edu.au, etzion@cs.technion.ac.il, hmkiah@ntu.edu.sg, luoyuan@cs.sjtu.edu.cn}
  \vspace{-5mm}
  \thanks{\IEEEauthorrefmark{5}Corresponding author: Wenqin Zhang.}
  \thanks{  The work was supported in part by National Key $R\&D$ Program of China under Grant 2022YFA1005000,  National Natural Science Foundation of China under Grant 62171279, and Fundings of SJTU-Alibaba Joint Research Lab on Cooperative Intelligent Computing. The work of Son Hoang Dau was supported by ARC DECRA DE180100768. The work of Han Mao Kiah was supported by the Ministry of Education, Singapore, under its MOE AcRF Tier~2 Award under Grant MOE-T2EP20121-0007 and MOE AcRF Tier~1 Award under Grant RG19/23.
}
  }
\date{}

\begin{document}

\maketitle

\hspace*{-3mm}\begin{abstract}
To measure repair latency at helper nodes, we introduce a new metric called {\em skip cost} that quantifies the number of contiguous sections accessed on a disk.
We provide explicit constructions of zigzag\,codes and fractional\,repetition\,codes that incur zero skip cost.
%\vspace{-2mm}
\end{abstract}

\section{Introduction}

In a distributed storage system, a file is represented as $\bfx\in\FF^k$ for a finite field $\FF$. 
To protect against node failures, the data is encoded into a codeword $\bfc = (c_1,\ldots,c_n)$, where $n$ is larger than $k$ and each codesymbol $c_i$ is kept in node $i$. Traditionally, erasure coding is employed, enabling the recovery of the file $\bfx$ by contacting a subset of these nodes. 
In the case of maximum distance-separable (MDS) codes, an optimal strategy exists, allowing one to recover $\bfx$ from any $k$~nodes.

However, in the scenario of a single node failure, this approach is suboptimal, as we download $k$ codesymbols to recover just one codesymbol. 
Therefore, in~\cite{dimakis2010network}, the notion of {\em repair bandwidth} was introduced as a metric to measure the performance of a distributed storage system. Essentially, the repair bandwidth quantifies the total amount of information downloaded from {\em helper nodes} to repair any failed node. Following this seminal work, numerous coding proposals emerged and we direct the reader to~\cite{dimakis2011survey,liu2018overview,ramkumar2021codes} for detailed surveys.

Apart from repair bandwidth, additional factors contribute to latency in the repair process.  For example, the notion of access-optimality was consider in~\cite{ye2017explicit} and~\cite{vajha2018clay}, while here we introduce another metric. Specifically, in~\cite{wu2021}, the access latency at the helper nodes is attributed to two factors: data access time and data processing time. Now, the latter was addressed in earlier works where the concept of {\em repair-by-transfer} was introduced~\cite{shah2012}.
In the repair-by-transfer framework, a failed node is repaired through a simple transfer of data without the need for processing at the helper nodes%
\footnote{In~\cite{shah2012}, a more stringent requirement is placed, where the failed node performs no computation. We relax this condition for this work, focusing on reducing the delay at the helper nodes.}.
Here, we investigate two classes of codes exhibiting this property: {\em zigzag\,codes}~\cite{Tamo2013} and {\em fractional\,repetition\,codes}~\cite{elrouayheb2010}. 
On a more significant note, \cite{wu2021} characterized data access time not only by the amount of data read (typically known as I/O costs), but also the number of contiguous sections accessed on a disk%
\footnote{In fact, this issue was raised during a Huawei Industry Session during ISIT 2023. Dr Wu discussed this in his presentation ``Coding Challenges for Future Storage''.}. In this work, we propose a novel metric named {\em skip cost} to quantify this notion.

To illustrate this metric, we consider the toy examples presented in Fig.~\ref{fig:zigzag}. In (a), we have five nodes $\bfa^{(0)}$, $\bfa^{(1)}$, $\bfa^{(2)}$, $\bfp^{(0)}$ and $\bfp^{(1)}$. 
When $\bfa^{(2)}$ fails, the contents highlighted in {\color{blue} blue} are read (see Section~\ref{sec:zigzag} for details of the repair procedure). Notably, in node $\bfa^{(0)}$, the read contents are not contiguous and we quantify the gap as skip cost one. Summing up the skip costs across all helper nodes, we say that the skip cost of the repair is four.
In contrast, (b) illustrates an example of our proposed code construction. In this case, the repair bandwidth remains the same, but we see that all read contents are contiguous. Consequently, the skip cost is zero.

In this work, we present explicit constructions of array codes with zero skip cost.
Before delving into technical descriptions of our contributions, we establish some notations.

\begin{figure}[!t]
\footnotesize
\begin{center}
    \noindent(a) $(4\times 5,3)$-array code with skip cost four.
    
    \vspace{1mm}
%\footnotesize
\setlength\tabcolsep{1.5pt} 
    \begin{tabular}{|c||c|c|c|c|c|}
    \hline
    & $\bfa^{(0)}$ & $\bfa^{(1)}$ & $\bfa^{(2)}$ & 
    $\bfp^{(0)}$ & $\bfp^{(1)}$ \\
    & $\spadesuit$  & $\heartsuit$ &  $\clubsuit$ &  
    $\bfS_0=(\vzero,\vzero,\vzero)$ & $\bfS_1=(\vzero,\bfe_1,\bfe_2)$\\ \hline 
    $00$ & {\color{blue}$00_\spadesuit$} & {\color{blue}$00_\heartsuit$}  & {\color{red}$00_\clubsuit$} & 
    {\color{blue}$00_{ \spadesuit }\boxplus 00_{\heartsuit }\boxplus 00_{\clubsuit}$} & 
    {\color{blue}$00_{ \spadesuit }\boxplus 10_{\heartsuit }\boxplus 01_{\clubsuit}$} \\ 
    
    $01$ & $01_\spadesuit$ & $01_\heartsuit$  & {\color{red}$01_\clubsuit$} &
    $01_{ \spadesuit }\boxplus 01_{\heartsuit}\boxplus 01_{\clubsuit}$ &
    $01_{ \spadesuit }\boxplus 11_{\heartsuit}\boxplus 00_{\clubsuit}$ \\ 
    
    $10$ & {\color{blue}$10_\spadesuit$} & {\color{blue}$10_\heartsuit$}  & {\color{red}$10_\clubsuit$} & 
    {\color{blue}$10_{ \spadesuit }\boxplus 10_{\heartsuit }\boxplus 10_{\clubsuit}$} &
    {\color{blue}$10_{ \spadesuit }\boxplus 00_{\heartsuit }\boxplus 11_{\clubsuit}$} \\ 
    
    $11$ & $11_\spadesuit$ & $11_\heartsuit$  & {\color{red}$11_\clubsuit$} & 
    $11_{ \spadesuit }\boxplus 11_{\heartsuit }\boxplus 11_{\clubsuit}$ & 
    $11_{ \spadesuit }\boxplus 01_{\heartsuit }\boxplus 10_{\clubsuit}$ \\ \hline
    \end{tabular}
\end{center}

    \vspace{2mm}
    \begin{center}
    \noindent(b) $(4\times 6,3)$-array code code with skip cost zero.
    \vspace{1mm}
    
{\centering \scriptsize
    \setlength\tabcolsep{1pt} 
    \begin{tabular}{|c||c|c|c|c|c|c|}
    \hline
    & $\bfa^{(0)}$ & $\bfa^{(1)}$ & $\bfa^{(2)}$ & 
      $\bfp^{(0)}$ & $\bfp^{(1)}$ & $\bfp^{(2)}$ \\
    & $\spadesuit$  & $\heartsuit$ &  $\clubsuit$ &  
    $\bfS_0=(\vzero,\vzero,\vzero)$ & $\bfS_1=(\vzero,\bfe_2,\bfd_1)$ & $\bfS_1=(\vzero,\bfd_1,\bfd_2)$ \\ \hline 
    $00$ & {\color{blue}$00_\spadesuit$} & {\color{blue}$00_\heartsuit$}  & {\color{red}$00_\clubsuit$} & 
    {\color{blue}$00_{ \spadesuit }\boxplus 00_{\heartsuit }\boxplus 00_{\clubsuit}$} & 
    {\color{blue}$00_{ \spadesuit }\boxplus 01_{\heartsuit }\boxplus 10_{\clubsuit}$} &
    $00_{ \spadesuit }\boxplus 10_{\heartsuit }\boxplus 11_{\clubsuit}$  \\ 
    
    $01$ & {\color{blue}$01_\spadesuit$} & {\color{blue}$01_\heartsuit$}  & {\color{red}$01_\clubsuit$} &
    {\color{blue}$01_{ \spadesuit }\boxplus 01_{\heartsuit}\boxplus 01_{\clubsuit}$} &
    {\color{blue}$01_{ \spadesuit }\boxplus 00_{\heartsuit}\boxplus 11_{\clubsuit}$} &
    $01_{ \spadesuit }\boxplus 11_{\heartsuit }\boxplus 10_{\clubsuit}$  \\ 
    
    $10$ & $10_\spadesuit$ & $10_\heartsuit$  & {\color{red}$10_\clubsuit$} & 
    $10_{ \spadesuit }\boxplus 10_{\heartsuit }\boxplus 10_{\clubsuit}$ &
    $10_{ \spadesuit }\boxplus 11_{\heartsuit }\boxplus 00_{\clubsuit}$ &
    $10_{ \spadesuit }\boxplus 00_{\heartsuit }\boxplus 01_{\clubsuit}$  \\ 
    
    $11$ & $11_\spadesuit$ & $11_\heartsuit$  & {\color{red}$11_\clubsuit$} & 
    $11_{ \spadesuit }\boxplus 11_{\heartsuit }\boxplus 11_{\clubsuit}$ & 
    $11_{ \spadesuit }\boxplus 10_{\heartsuit }\boxplus 01_{\clubsuit}$ &
    $11_{ \spadesuit }\boxplus 01_{\heartsuit }\boxplus 00_{\clubsuit}$  \\  \hline
    \end{tabular}}
    \end{center}
    \caption{(a)~Example of a {$(4\times 5,3)$}-MDS array code constructed in~\cite{Tamo2013}. Suppose information node $\bfa^{(2)}$ (highlighted in {\color{red}red}) fails. We contact nodes $\bfa^{(0)}$, $\bfa^{(1)}$, $\bfp^{(0)}$ and $\bfp^{(1)}$ and read the contents in {\color{blue}blue}. Here, the skip cost is $4\times 1 = 4$. 
    (b)~Example of a {$(4\times 6,3)$}-MDS array code described by Construction~A with $m=2$ (see Section~\ref{sec:zigzag}). Suppose information node $\bfa^{(2)}$ (highlighted in {\color{red}red}) fails. We contact nodes $\bfa^{(0)}$, $\bfa^{(1)}$, $\bfp^{(0)}$ and $\bfp^{(1)}$ and read the contents in {\color{blue}blue}. Here, the skip cost is zero.
    Note that we use $\bfx_{i}$ to represent the information symbol $a_\bfx^{(i)}$, while the `sum' $\bfx_i\boxplus \bfy_j\boxplus \bfz_k$ indicates that the corresponding codesymbol is a linear combination of $a_\bfx^{(i)}$, $a_\bfy^{(j)}$, and $a_\bfz^{(k)}$. %\hm{Use plus.}
    %\vspace{-5mm}
    }\label{fig:zigzag}
\end{figure}

\section{Problem Statement}

For integers $i$ , $j$ with $i<j$, we let $[i,j]$ and $[i]$ denote the sets of integers $\{i,i+1,i+2,\ldots, j\}$ and $\{1,2,\ldots, i\}$, respectively. Furthermore, $\FF_q$ denotes the finite field with $q$.

Throughout this paper, we consider {\em array codes}, where each codeword is represented as an $M\times N$-array over $\FF_q$, whose rows and columns are indexed by $[M]$ and $[N]$, respectively.
Specifically, each codeword is of the form $\bfC = \left[\bfb^{(1)}, \bfb^{(2)},\ldots, \bfb^{(N)}\right]$, 
where $\bfb^{(j)}=\left(b_1^{(j)},\ldots, b_M^{(j)}\right)$ belongs to $\FF_q^M$.
Here, we say that we have $N$ {\em nodes} with each node storing $M$ {\em packets} or {\em code symbols}.

If the codewords form an $\FF_q$-linear subspace of dimension $Mk$, we say that the code is an {\em $(M\times N, k)$-array code}. 
Consider any codeword $\bfC$. If we can recover all $N$ columns of $\bfC$ from any subset of $k$ columns, we say the code is {\em maximum distance separable (MDS)}.

\begin{comment}
It was proved in \cite{dimakis2010network} that for an $[n, k]$-MDS code with code
length $n$ and dimension $k$, the recovery of a single failed node
from any $d$ helper nodes should download at least a fraction $\frac{1}{d-k+1}$
of the data stored in each of the helper nodes, i.e., the repair
bandwidth $\gamma(d)$ satisfies
\begin{equation}\label{eq:repair}
     \gamma(d)\geq \frac{d}{d-k+1}N
\end{equation}    
\end{comment}

When a node containing $\bfb$ fails, we contact certain {\em helper nodes} $\cH$ and download a certain portion of their contents to {\em repair} the failed node, that is, to recover $\bfb$. 
The amount of information downloaded from the helper nodes is called the {\em repair bandwidth}.
For an $(M\times N, k)$-array code, \cite{dimakis2010network} demonstrated that if $|\cH|=d$, then it is necessary to download at least $1/(d-k+1)$ fraction of the contents stored on each helper node.
This lower bound is commonly referred to as the {\em cutset bound} and most work seeks to achieve this bound.

In addition to this bound, this paper considers the {\em repair-by-transfer} framework, where the content downloaded is precisely the information read from each helper node. Besides minimizing the amount of information read, we introduce another metric, called {\em skip cost}.

\begin{definition}[Skip Cost]
Consider a tuple $\bfb=(b_1,b_2,\ldots,b_M)$ and 
suppose that we read $t$ packets $\Read=\{b_{i_1},\ldots, b_{i_t}\}$ where $i_1<\cdots < i_t$.
We define the {\em skip cost} of $\Read$ with respect to $\bfb$ is $\cost(\Read,\bfb)\triangleq i_t-i_1-(t-1)$. 
So, if the indices in $\Read$ are consecutive, then its skip cost is zero, which is clearly optimal.

An array code has a repair scheme with {\em skip cost $\sigma$} if for any codeword node $\bfb$,
there exists a collection of helper nodes $\cH$ and corresponding reads $\Read_\bfh~(\bfh\in\cH)$ such that
\begin{enumerate}[(i)]
    \item we can determine $\bfb$ from $\bigcup_{\bfh\in\cH} \Read_h$, and
    \item the total skip cost is at most $\sigma$, that is, $\sum_{\bfh\in\cH} \cost(\Read_h,\bfh) \le \sigma$.
\end{enumerate}
\end{definition}

In this paper, we consider two classes of array codes with good repair-by-transfer schemes: {\em zigzag codes} and {\em fractional repetition codes}, proposed in~\cite{Tamo2013} and~\cite{elrouayheb2010}, respectively. 
%\hm{to be continued. talk about optimal rebuilding ratio.}

\subsection{Zigzag Codes}

Zigzag codes are a special class of MDS array codes designed by Tamo {\em et al.}~\cite{Tamo2013} that have an optimal rebuilding ratio. Specifically, these codes are $(M\times N,k)$-MDS array codes where the first $k$ columns or nodes are systematic nodes and we focus only the failure of {\em single systematic nodes}. The {\em rebuilding ratio} of a code is then given by the maximum fraction of contents downloaded from a helper node to repair any systematic node.
Following the cutset bound, when the number of helper nodes is $d$, then the rebuilding ratio is necessarily at least $1/(d-k+1)$. Formally, we have the following result from~\cite{Tamo2013}.

\begin{theorem}[{\hspace*{-0.5mm}\cite{Tamo2013}}]
    Fix $r$. For a sufficiently large field, there exists an $(M\times N, k)$-MDS array code with $N-k=r$ and an optimal rebuilding ratio $1/r$. Here, the number of helper nodes is $N-1$.
\end{theorem}

Now, as we see in Section~\ref{sec:zigzag} (see Remark~\ref{rem:tamo-cost}), the repair scheme for these codes incurs a significant skip cost. Hence, in the same section, we design a new class of zigzag codes that incurs zero skip cost. Unfortunately, our array codes suffer a reduction in coding rates. Nevertheless, our codes retain valuable properties such as optimal rebuilding ratio and optimal update%
\footnote{As defined in~\cite{Tamo2013}, a code has optimal update if whenever an information symbol is altered, only the information symbol itself and one symbol from each parity node are updated.}.

\subsection{Fractional Repetition Codes}
\label{sec:prelim-fr}

\begin{figure}[!t]
\footnotesize
\begin{center}
    \noindent(a) $(4\times 14)$-array code with locality two and  skip cost two.
    \vspace{1mm}
    
 %\footnotesize
    \begin{tabular}{|c|c|c|c|c|c| c|c|c|c|c|c| c|c|c|c|c|}
    \hline
    
    1 & 1 & 1 & 1 &
    1 & 1 & 1 & 5 &
    3 & 3 & {\color{red}2} & 7 &
    3 & 3 \\ 
    
    {\color{blue}2} & 2 & 2 & 3 &
    {\color{blue}6} & 4 & 4 & 6 &
    4 & 4 & {\color{red}4} & 4 &
    2 & 2 \\ 

    3 & 5 & 7 & 5 &
    3 & 5 & 7 & 7 &
    7 & 5 & {\color{red}6} & 5 &
    7 & 5 \\

    {\color{blue}4} & 6 & 8 & 7 &
    {\color{blue}8} & 8 & 6 & 8 &
    8 & 6 & {\color{red}8} & 2 &
    6 & 8 \\
     \hline
    \end{tabular}
\end{center}

    \vspace{2mm}
     \begin{center}
    \noindent(b) $(4\times 14)$-array code with locality two and skip cost zero.
    \vspace{1mm}
    
     \setlength\tabcolsep{4pt} 
    \begin{tabular}{|c|c|c|c|c|c| c|c|c|c|c|c| c|c|c|c|c|}
    \hline
    \multicolumn{8}{|c|}{$\cB_1$} & 
    \multicolumn{6}{ c|}{$\cB_2$} \\\hline
    {\color{red}$1_0$} & {\color{blue}$1_0$} & $1_0$ & $1_0$ &
    $1_1$ & $1_1$ & $1_1$ & $1_1$ &
    $1_0$ & $1_0$ & $1_0$ & $2_0$ & $2_0$ & $3_0$ \\

    {\color{red}$2_0$} & {\color{blue}$2_0$} & $2_1$ & $2_1$ &
    $2_0$ & $2_0$ & $1_1$ & $2_1$ &
    $1_1$ & $1_1$ & $1_1$ & $2_1$ & $2_1$ & $3_1$ \\
    
    {\color{red}$3_0$} & $3_1$ & $3_0$ & $3_1$ &
    $3_0$ & $3_1$ & {\color{blue}$3_0$} & $3_1$ &
    $2_0$ & $3_0$ & $4_0$ & $3_0$ & $4_0$ & $4_0$ \\

    {\color{red}$4_0$} & $4_1$ & $4_1$ & $4_0$ &
    $4_1$ & $4_0$ & {\color{blue}$4_0$} & $4_1$ &
    $2_1$ & $3_1$ & $4_1$ & $3_1$ & $4_1$ & $4_1$ \\
    
     \hline
    \end{tabular}
 \end{center} %\footnotesize
     
    \caption{(a)~Example of a {$(4\times 14)$}-array code built using a $\sqs(8)$. 
    Suppose the node highlighted in {\color{red}red} fails. We read the contents in {\color{blue}blue} and the skip cost is $2\times 1 = 2$. 
    (b)~Example of a {($4\times 14$)}-array code described by Construction~D with $|V|=8$ (see Section~\ref{sec:fractional}). Here, due to the space constrained,  we use the notation $a_b$ to denote $(a,b)$. }Suppose the node highlighted in {\color{red}red} fails. We read the contents in {\color{blue}blue} and the skip cost is zero.
    %\vspace{-5mm}
    \label{fig:sqs}
\end{figure}

After the introduction of the repair-by-transfer framework in~\cite{shah2012}, El~Rouayheb and Ramchandran~\cite{elrouayheb2010} introduced a novel class of array codes that uses table-based repair. 
Known as DRESS (Distributed Replication based Exact Simple Storage) codes, these codes comprise two main components: an outer MDS code and an inner repetition code known as {\em fractional repetition} code. Typically, fractional repetition codes are constructed using combinatorial designs, and in-depth discussions can be found in~\cite{elrouayheb2010,ernvall2012existence, olmez2016fractional, zhu2014general, silberstein2015optimal}.

Here, we formally describe how a set system can be used to construct an array code.
Let $V$ be a set of {\em points} and $\cB$ be a collection of subsets defined over $V$. 
For $\bfb\in\cB$, we refer to $\bfb$ as a {\em block}.
Then $(V,\cB)$ is known as a {\em set system}.

Suppose further that $|V|=n$, $|\cB|=N$ and $|\bfb|=M$ for all $\bfb\in\cB$.
In addition, we also consider an $[n,k]$-MDS outer code%
\footnote{In the rest of the paper,  we omit specifying the dimension $k$ of the outer code, as it does not influence the repair costs.}.
In a DRESS code, a file ${\bfx}$ is broken up into $n$ {\em coded packets} and stored in $|\cB|=N$ nodes such that any $k$ packets suffices to restore ${\bfx}$. 
Here, the packets of the file ${\bfx}$ and the storage nodes are indexed by $V$ and $\cB$, respectively.
Then for a block $\bfb\in \cB$, we store the packets $\{{\bfx}_{v}: v\in\bfb\}$ in the node $\bfb$ and 
we obtain an $(M\times N)$-array code.
For example, the $(4\times 14)$-array codes displayed in Fig.~\ref{fig:sqs}
are constructed using set systems with eight points, $14$ blocks, and with constant block size four. 

As before, we are interested in the event when a single node fails.
In the case of DRESS codes, the repair process for the failed node or block involves downloading the coded packets replicated in other nodes. To further reduce repair latency, we limit the number of helper nodes, imposing locality requirements on the repair process.

Now, given certain locality requirements, we observe that the {\em order} in which the points or packets are arranged in each node has an impact on skip cost. 
In Fig.~\ref{fig:sqs}(a), when the node $(2,4,6,8)$ fails, a skip cost of two is necessary to recover the points $2$, $4$, $6$, and $8$ from two helper nodes. In contrast, with the arrangement of points in  Fig.~\ref{fig:sqs}(b), we can achieve both locality two and skip cost zero, regardless of the failed node.
For convenience, we say that a set system has {\em locality $d$} and {\em skip cost $\sigma$} if the corresponding array code admits a repair scheme that repairs with at most $d$ helpers and skip cost $\sigma$.

Next, we observe that both set systems in Figure~\ref{fig:sqs} are Steiner quadruple systems (SQS) of order eight, and these designs shall be our main object of study%
\footnote{A SQS is a special class of three-designs as each triple appears exactly once among all blocks. Now, a two-design is defined to be a set system where every pair appears exactly once among all blocks. In this case, during recovery, at most one point can be downloaded from every helper node. Then the skip cost is trivially zero. Therefore, we study the next non-trivial case: SQS. }.

\begin{definition}
    A {\em Steiner quadruple systems} of order $v$, or $\sqs(v)$, is a set system $(V,\cB)$ such that
    \begin{enumerate}[(i)]
    \item $|V|=v$ and $|\bfb|=4$ for all $\bfb\in\cB$.
    \item Every subset $T$ of $V$ with $|T|=3$ is contained in exactly one block in $\cB$.
    \end{enumerate}
\end{definition}

Hanani then demonstrated the existence of such quadruple systems for all admissible parameters.

\begin{theorem}[Hanani~\cite{hanani1960quadruple}]\label{thm:hanani}
A $\sqs(v)$ exists if and only if $v\ge 4$ and $v \equiv 2$ or $4\pmod{6}$.
\end{theorem}

In Section~\ref{sec:fractional}, we show that at least two thirds of these quadruple systems have locality two and skip cost zero.
Specifically, we show the following result.

\begin{theorem}\label{thm:sqs-all}
There exists a $\sqs(v)$ with locality two and skip cost zero if
$v\ge 4$ and $v\equiv 4,8,10,16,20,22,28,\text{ or }32\pmod{36}$
and $v\in\{14,26,34,38\}$.
\end{theorem}
Therefore, we have an $\sqs(v)$ with locality two and skip cost zero for all admissible parameters at most 50. 
We conjecture that a $\sqs(v)$ with locality two and skip cost zero exists for all admissible parameters.
%\hm{We may update this paragraph if we obtain more results.}
%and a combinatorial design is a set system with certain intersection properties.

\begin{comment}
\begin{itemize}
    \item When we treat each block as an {\em unordered} set, we refer to the pair $(V,\cB)$ as an {\em unordered set system}.
    \item When the order in each block matters, i.e. $\bfb$ is a tuple, we refer to the pair $(V,\cB)$ as an {\em ordered set system}. Here, we label the indices of $\bfb$ with $[|\bfb|]$. In other words, $\bfb=(v_1,v_2,\ldots, v_{|\bfb|})$.
\end{itemize}
\end{comment}

\begin{comment}
 A $t$-$(N,M,\lambda)$-{\em design} is a unordered set system $(V,\cB)$ such that
\begin{enumerate}[(i)]
\item $|V|=n$ and $|\bfb|=M$ for all $\bfb\in\cB$.
\item For all $t$-subsets $T$ of $V$, the subset $T$ appears exactly $\lambda$ times in all blocks in $\cB$.   
\end{enumerate}
\end{comment}

%\section{Zigzag Codes {\em \'a la} Tamo-Wang-Bruck}
\section{Zigzag Codes with Zero Skip Cost}
\label{sec:zigzag}

In this section, we present three classes of MDS array codes with a repair scheme that incurs zero skip cost. Our initial construction, Construction~A, serves to demonstrate the main ideas underlying the placement of parity symbols. Subsequently, we refine our approach. Specifically, we achieve a higher code rate in Construction~B and later eliminate the dependency of number of information symbols and sub-packetization level in Construction~C.

Although our placement of parity code symbols shares similarities with that of Tamo-Wang-Bruck, 
specific differences enable us to achieve zero skip cost during the repair process.
\vspace{1mm}

\noindent{\bf Construction A.} 
Let $m\geq2$ and set $k=m+1$. 
We present an $(M\times N,k)$-MDS array code with $M=2^m$ packets and $N=2k=2(m+1)$ nodes.
Of these $N$ nodes, $k=m+1$ of them are {\em systematic} nodes $\bfa^{(0)},\bfa^{(1)}\ldots, \bfa^{(m)}$, while the remaining $k=m+1$ are parity nodes $\bfp^{(0)}, \bfp^{(1)},\ldots, \bfp^{(m)}$.
Now, instead of indexing the rows with $[M]$, we index them using the $M=2^m$ bitstrings in $\FF_2^m$. 
Here, the rows are arranged in lexicographic order.
Moreover, for convenience, we introduce the following notation for certain vectors in $\FF_2^m$:
\begin{itemize}
    \item $\vzero$ denotes the zero vector;
    \item For $i\in [m]$, we use $\bfe_{i}$ to denote the vector whose $i$-th entry is one and other entries are $0$.
    \item For $i\in [m]$, we use $\bfd_{i}$ to denote the vector whose first $i$ entries are one and other entries are $0$. In other words, $\bfd_i=\bfe_1+\cdots+\bfe_i$.
    \item Finally, we use $\cU$ and $\cL$ denotes the set of bitstrings starting with zero and one, respectively. Given the lexicographic order, we observe that the first and second halves of the array are indexed by strings in $\cU$ and $\cL$, respectively.
\end{itemize}

\vspace{1mm}
\noindent\textit{Contents of systematic columns}: For $j\in[0,m]$, we simply set $\bfa^{(j)}=\left(a^{(j)}_{\bfx}\right)_{\bfx\in\FF_2^m}$.

\vspace{1mm}

\noindent\textit{Contents of parity columns}: For $j\in[0,m]$, we set $\bfp^{(j)}=\left(p^{(j)}_{\bfx}\right)_{\bfx\in\FF_2^m}$ and our task is to determine the parity sum $p^{(j)}_{\bfx}$ in terms of the $a^{(j)}_\bfx$'s. 
Now, in what follows, each parity sum in $\bfp^{(j)}$ is a linear combination $m+1$ information symols and is determined by an $(m+1)$-tuple $\bfS_j$. %Specifically, we do the following.

\begin{itemize}\itemindent=-10pt
\item Set $\bfS_0$ to be $(\vzero,\ldots, \vzero)$ and 
for $j\in [m]$, we set $S_j=\left(\vzero, \bfe_{j+1},\bfe_{j+2},\ldots,\bfe_m,\bfd_1,\bfd_2,\ldots,\bfd_j\right)$.
\item Now, we use the tuple $\bfS_j$ to generate a parity column $\bfp^{(j)}$. 
If $\bfS_j=(\bfv_0,\ldots, \bfv_m)$ and we consider row $\bfx$, we set $\bfu_i=\bfx+\bfv_i$ and define $p^{(j)}_{\bfx} = \sum^{m}_{i=0}\alpha^{(i,j)}_{\bfx}a^{(i)}_{\bfu_i}$ for some choice of coefficients $\alpha^{(i,j)}_{\bfx}\in\FF_q$ 
\item For example, for the first parity column $\bfp^{(0)}$, we have that $\bfS_0=(\vzero,\ldots, \vzero)$.
Then for row $\bfx$, the parity sum is $p^{(0)}_{\bfx}=\alpha^{(0,0)}_{\bfx}a^{(0)}_{\bfx}+\cdots+\alpha^{(m,0)}_{\bfx}a^{(m)}_{\bfx}$.
In other words, $p^{(0)}_{\bfx}$ is a linear combination of the $(m+1)$ information symbols in the same row.
\end{itemize}

\vspace{1mm}
\noindent\textit{Repair of Information Node $s$ for $s\in [0,m]$}:
\begin{itemize}\itemindent=-10pt
\item If $s\in [m]$, set helper nodes to be $\cH = \{\bfa^{(i)}: i\in [0,m]\setminus\{s\}\}\cup\{\bfp^{(m-s)},\bfp^{(m-s+1)}\}$. 
\begin{itemize}\itemindent=-20pt
    \item From $\bfa^{(i)}$ with $i<s$, we read $\Read_{\bfa^{(i)}}=\{a^{(i)}_\bfx:\bfx\in\cU\}$.
    \item From $\bfa^{(i)}$ with $i>s$, we read $\Read_{\bfa^{(i)}}=\{a^{(i)}_\bfx:\bfx\in\cL\}$.
    \item From $\bfp^{(i)}$ with $i\in\{m-s,m-s+1\}$, we read $\Read_{\bfp^{(i)}}=\{p^{(i)}_\bfx:\bfx\in\cU\}$.
\end{itemize}
\item If $s=0$, set helper nodes to be $\cH = \{\bfa^{(i)}: i\in [m]\}\cup\{\bfp^{(0)},\bfp^{(m)}\}$. 
\begin{itemize}\itemindent=-20pt
    \item From $\bfa^{(i)}$ with $i\in [m]$, we read $\Read_{\bfa^{(i)}}=\{a^{(i)}_\bfx:\bfx\in\cU\}$.
    \item From $\bfp^{(0)}$, we read $\Read_{\bfp^{(0)}}=\{p^{(0)}_\bfx:\bfx\in\cU\}$.
    \item From $\bfp^{(m)}$, we read $\Read_{\bfp^{(m)}}=\{p^{(m)}_\bfx:\bfx\in\cL\}$.
\end{itemize}
\end{itemize}

\begin{theorem}\label{thm:zigzag1}
The repair scheme for Construction~A is correct. 
That is, for $s\in [0,m]$, we are able to determine $\bfa^{(s)}$ from all reads.
Moreover, the repair scheme has zero skip cost and optimal rebuilding ratio $1/2$.
\end{theorem}

\begin{proof}
We prove for the case $s\in[m]$ and the cases $s=0$ can be similarly proved.

Recall that the reads from $\bfp^{(m-s)}$ and $\bfp^{(m-s+1)}$ are
\begin{align*}
\Read_{\bfp^{(m-s)}}   &= \left\{ \sum_{i=0}^m \alpha_\bfx^{(i,j)} a^{(i)}_{\bfu_i}: \bfx\in \cU\right\},\\
\Read_{\bfp^{(m-s+1)}} &= \left\{ \sum_{i=0}^m \alpha_\bfx^{(i,j)} a^{(i)}_{\bfu_i}: \bfx\in \cU\right\}.    
\end{align*}
Here, let $\cI_1$ and $\cI_2$ be the information symbols involved in the reads $\Read_{\bfp^{(m-s)}}$ and $\Read_{\bfp^{(m-s+1)}}$, respectively.

For $s\in [m-1]$, we observe that 
\begin{align*}
\bfS_{m-s} &= (\vzero,\bfe_{m-s+1},\ldots, {\color{blue}\bfe_m}, \ldots, \bfd_{m-s-1}),\\
\bfS_{m-s+1} &= (\vzero,\bfe_{m-s+2},\ldots, {\color{blue}\bfd_1}, \ldots, \bfd_{m-s}),     
\end{align*}
for $s=m$, $
\bfS_{0} = (\vzero,\vzero,\ldots, \vzero,{\color{blue}\vzero})$, $
\bfS_{1} = (\vzero,\bfe_{2},\ldots,  \bfe_m,{\color{blue}\bfd_1})$,
where the highlighted positions in blue are the $s$-th entries. Using this observation, we claim the following.
\begin{itemize} 
\setlength{\leftmargin}{0pt}
\item {\em $\cI_1\cup \cI_2$ contains all information symbols $a^{(s)}_\bfx$ with $\bfx\in\FF_2^m$}. Indeed, if $\bfx\in \cU$, since $\bfx+\bfe_m\in \cU$, we have that $a^{(s)}_\bfx\in \cI_1$. Similarly, if $\bfx\in \cL$, since $\bfx+\bfd_1\in \cU$, we have that $a^{(s)}_\bfx\in \cI_2$.

\item {\em If $i<s$, then $\cI_1\cup \cI_2$ contains exactly all information symbols $a^{(i)}_\bfx$ with $\bfx\in\cU$}. Indeed, in both $\bfS_{m-s}$ and $\bfS_{m-s+1}$, the $i$-th entries belong to $\cU$. Hence, we have that $a^{(i)}_{\bfu_i}$ always belong to $\cU$.

\item {\em If $i>s$, then $\cI_1\cup \cI_2$ contains exactly all information symbols $a^{(i)}_\bfx$ with $\bfx\in\cL$}. As before, we observe that in both $\bfS_{m-s}$ and $\bfS_{m-s+1}$, the $i$-th entries belong to $\cL$. Hence, we have that $a^{(i)}_{\bfu_i}$ always belong to $\cL$.
\end{itemize}

Therefore, using the reads $\Read_{\bfa^{(i)}}~(i\ne s)$, we eliminate all information symbols $a^{(i)}_\bfx$ with $i\ne s$ from the corresponding check sums and recover $\bfa^{(s)}$.

For the case of $s=0$, the reads from $\bfp^{(0)}$ and $\bfp^{(m)}$ are
\begin{align*}
\Read_{\bfp^{(0)}}   &= \left\{ \sum_{i=0}^m \alpha_\bfx^{(i,j)} a^{(i)}_{\bfu_i}: \bfx\in \cU\right\},\\
\Read_{\bfp^{(m)}} &= \left\{ \sum_{i=0}^m \alpha_\bfx^{(i,j)} a^{(i)}_{\bfu_i}: \bfx\in \cL\right\}.    
\end{align*}
Here, let $\cI_1$ and $\cI_2$ be the information symbols involved in the reads $\Read_{\bfp^{(0)}}$ and $\Read_{\bfp^{(m)}}$, respectively.

For $s=0$, we observe that 
\begin{align*}
\bfS_{0} &= (\vzero,\vzero,\ldots, \vzero),\\
\bfS_{m} &= (\vzero, \bfd_1, \ldots, \bfd_{m-s}),     
\end{align*}
% where the highlighted positions in blue are the $s$-th entries. Using this observation, we claim the following.
\begin{itemize} 
\setlength{\leftmargin}{0pt}
\item {\em Obviously, $\cI_1\cup \cI_2$ contains all information symbols $a^{(0)}_\bfx$ with $\bfx\in\FF_2^m$}. 
% Indeed, $\bfx\in \cU \cup\bfx\in \cL$ , since $\bfx+\bfe_m\in \cU$, we have that $a^{(s)}_\bfx\in \cI_1$. Similarly, if $\bfx\in \cL$, since $\bfx+\bfd_1\in \cU$, we have that $a^{(s)}_\bfx\in \cI_2$.
\item {\em If $i>0$, then $\cI_1\cup \cI_2$ contains exactly all information symbols $a^{(i)}_\bfx$ with $\bfx\in\cU$}. Indeed, in both $\bfS_{0}$ and $\bfS_{m}$, the $i$-th entries belong to $\cU$. Hence, we have that $a^{(i)}_{\bfu_i}$ always belong to $\cU$.
\end{itemize}
Therefore, using the reads $\Read_{\bfa^{(i)}}~(i\ne s)$, we eliminate all information symbols $a^{(i)}_\bfx$ with $i\ne s$ from the corresponding check sums and recover $\bfa^{(s)}$.

Finally,  to check that skip cost is zero, we observe that we always download symbols in all rows in either $\cU$ or $\cL$, which are consecutive. This also means that the rebuilding ratio is $1/2$.
This is optimal as the number of helper nodes is $m+2=k+1$.
\end{proof}

We illustrate Construction~A with an example.

\begin{example}
    Let $m=2$. Hence, $M=4$, $N=6$ with three information nodes and three parity nodes. 
    We present the placement of the codesymbols in Fig.~\ref{fig:zigzag}(b). 
    To reduce clutter and confusion, we index the information nodes with the symbols $\{\spadesuit,\heartsuit,\clubsuit\}$ in place of $\{0,1,2\}$. Additionally, instead of writing $a^{(i)}_\bfx$, we write it as $\bfx_{i}$. In other words, $00_{\spadesuit}$ represents the information symbol $a^{(\spadesuit)}_{00}$. For the parity symbols, we use $\bfx_i \boxplus \bfy_j \boxplus \bfz_k$ to represent a linear combination of $a_\bfx^{(i)}$, $a_\bfy^{(j)}$, and $a_\bfz^{(k)}$.

    We illustrate the repair scheme assuming that the information node corresponding to $\heartsuit$ (or $s=1$) has failed.
    According to our scheme, the helper information nodes are $\spadesuit$ and $\clubsuit$, while the helper parity nodes are $\bfp^{(1)}$ and $\bfp^{(2)}$. Specifically, the reads from the parity nodes contain the following information symbols.
    \begin{itemize}
    \item From $\bfp^{(1)}$, we obtain
    \[ \cI_1 \triangleq 
    \{00_{\spadesuit }, {\color{red}01_{\heartsuit}},10_\clubsuit\} \cup
    \{01_{\spadesuit }, {\color{red}00_{\heartsuit}},11_\clubsuit\}\,. \]
    \item From $\bfp^{(2)}$, we obtain
    \[ \cI_2 \triangleq 
    \{00_{\spadesuit }, {\color{red}10_{\heartsuit}},11_\clubsuit\} \cup
    \{01_{\spadesuit }, {\color{red}11_{\heartsuit}},10_\clubsuit\}\,. \]
    \end{itemize}
    Indeed, we check that 
    \[\cI_1\cup\cI_2 = \{\bfx_{\spadesuit}: \bfx\in\cU\}\cup\{\bfx_{\heartsuit}: \bfx\in{\color{red}\FF_2^2}\}\cup\{\bfx_{\clubsuit}: \bfx\in\cL\}\,,\]
    which corroborates with the claims in the proof of Theorem~\ref{thm:zigzag1}.
    Since we download the information symbols $\{\bfx_{\spadesuit}: \bfx\in\cU\}$ and $\{\bfx_{\clubsuit}: \bfx\in\cL\}$ from $\bfa^{(0)}$ and $\bfa^{(2)}$, respectively, we recover all values in $\{\bfx_{\heartsuit}: \bfx\in{\color{red}\FF_2^2}\}$, i.e., $\bfa^{(1)}$.
    Since we download contents from rows in either $\cU$ or $\cL$, the skip cost is zero.
    In Fig.~\ref{fig:zigzag}(b), we illustrate the repair scheme of another information node corresponding to $\clubsuit$ and  the skip cost is zero too. \qedsymbol
\end{example}

It remains to choose the coefficients $\alpha^{(i,j)}_\bfx$ so as to achieve the MDS property.
As the proof is technical and similar to that in~\cite{Tamo2013}, we defer the proof of the following proposition to Appendix~\ref{app:mds}. %\hm{Wenqin, can you write the proof?}

%{\color{red}
\begin{proposition}\label{prop:mds}
Fix $m$ and set $t=\lceil\frac{k}{2}\rceil$. 
%For sufficiently large $q> 2^m\binom{k-1}{t-1}^2$,
If $q> 2^m\binom{k-1}{t-1}^2$, then
there exist coefficients $\alpha^{(i,j)}_\bfx$ that belong to $\FF_q$ for $i,j\in [0,m]$ and $\bfx\in\FF_2^m$ 
such that we can recover $\bfa^{(i)}~(i\in[0,m])$ from any $k=m+1$ nodes.
\end{proposition}

\begin{remark}\label{rem:tamo-cost}
As mentioned earlier, the class of zigzag codes was proposed in~\cite{Tamo2013}.
Of significance, we compare Construction~A with Construction~1 in~\cite{Tamo2013}, which also have a rebuilding ratio $1/2$.
Similar to our construction, Construction~1 results in an array code with $2^m$ rows and $(m+1)$ information nodes.
However, Construction~1 uses only two parity check nodes $\bfp^{(0)}$ and $\bfp^{(1)}$ with corresponding $\bfS_0=(\vzero,\ldots,\vzero)$ and $\bfS_1=(\vzero,\bfe_1,\ldots,\bfe_m)$ (see Fig.~\ref{fig:zigzag}(a)). 
Notably, Construction~1 incurs a higher skip cost as illustrated in Fig.~\ref{fig:zigzag}(a).
Generally, when $m$ is even, the skip cost is 
$(m+2)\left(\sum^{m-1}_{i=0}2^{i}(2^{m-1-i}-1)+2^{m-1}\right)$;
while when $m$ is odd, the skip cost is
$(m+2)\left(\sum^{m-1}_{i=0}2^{i}(2^{m-1-i}-1)+2^{m-1}-1\right)$.
In both cases, the skip cost is at least $m^2 2^{m-1}$.

%Generally, for $m\ge 1$, when $2\mid m$, the skip cost is $(m+2)(2^{m-2}+2^{m-3}(2^2-1)+\ldots+(2^{m-1}-1))=(m+2)(\sum^{m-1}_{i=0}2^{i}(2^{m-1-i}-1)+2^{m-1})$; when $2\nmid m$, the skip cost is $(m+2)(2^{m-1}-1+2^{m-2}+2^{m-3}(2^2-1)+\ldots+(2^{m-1}-1))=(m+2)(\sum^{m-1}_{i=0}2^{i}(2^{m-1-i}-1)+2^{m-1}-1)$;  \hm{Wenqin, is this correct?}

In contrast, Construction~A has skip cost zero, but uses a significantly higher number of redundant nodes. 
It remains open whether we can achieve higher rates with zero skip cost and the same rebuilding ratio.
\end{remark}
%$$-----------$$

\vspace{1mm}

%{\color{red} Similar to Construction A, we introduce another class of MDS array codes featuring a repair scheme incurring zero skip cost. We will demonstrate the correctness of its repair scheme and establish that its code rate is approximately $2/3$, surpassing that of Construction A. Finally, we illustrate the construction's intuition through an example.}

Next, we present another class of MDS array codes with zero skip cost.
Furthermore, the code rate is approximately $2/3$ for large values of $k$.

\vspace{2mm}
\noindent{\bf Construction B.} 
Let $m\geq2$ and set $k=m+1$. 
We present an $(M\times N,k)$-MDS array code with $M=2^m$ packets and $N=k+ \left\lceil\frac{k}{2} \right\rceil+1=m+2+\left\lceil\frac{m+1}{2} \right\rceil$ nodes. Let $m'= \left\lceil\frac{m+1}{2} \right\rceil$. 
Of these $N$ nodes, $k=m+1$ of them are {\em systematic} nodes $\bfa^{(0)},\bfa^{(1)}\ldots, \bfa^{(m)}$, while the remaining $m'$ are parity nodes $\bfp^{(0)}, \bfp^{(1)},\ldots, \bfp^{(m')}$.
Now, instead of indexing the rows with $[M]$, we index them using the $M=2^m$ bitstrings in $\FF_2^m$. 
Here, the rows are arranged in lexicographic order.
Moreover, for convenience, we introduce the following notation for certain vectors in $\FF_2^m$:
\begin{itemize}
    % \item $\vzero$ denotes the zero vector;
    % \item For $i\in [m]$, we use $\bfe_{i}$ to denote the vector whose $i$-th entry is one and other entries are $0$.
    % \item For $i\in [m]$, we use $\bfd_{i}$ to denote the vector whose first $i$ entries are one and other entries are $0$. In other words, $\bfd_i=\bfe_1+\cdots+\bfe_i$.
    %\item {\color{red} Keep the above notation in Construction A. We use $\cU$ and $\cL$ denotes the set of bitstrings starting with zero and one respectively, $\cU'$ denotes the set of bitstrings starting with $``00"$  and $``11"$ and $\cL'=\FF_2^m  \setminus \cU'$. Given the lexicographic order, we observe that the first and second halves of the array are indexed by strings in $\cU$ and $\cL$, respectively. In addition, the first quarter and the last quarter of the array are indexed by strings in $\cU'$, the middle of the array is indexed by strings in $\cL'$.}
    \item As with Construction A, we use the notation $\cU$ and $\cL$ for the set of bitstrings starting with zero and one, respectively.
    Furthermore, we use $\cU'$ denotes the set of bitstrings starting with $``00"$  and $``11"$ and $\cL'=\FF_2^m  \setminus \cU'$. Given the lexicographic order, we observe that the first and second halves of the array are indexed by strings in $\cU$ and $\cL$, respectively. In addition, the first quarter and the last quarter of the array are indexed by strings in $\cU'$, the middle of the array is indexed by strings in $\cL'$.
\end{itemize}

\vspace{1mm}
\noindent\textit{Contents of systematic columns}: For $j\in[0,m]$, we simply set $\bfa^{(j)}=\left(a^{(j)}_{\bfx}\right)_{\bfx\in\FF_2^m}$.
\vspace{1mm}

\noindent\textit{Contents of parity columns}: For $j\in[0,m']$, we set $\bfp^{(j)}=\left(p^{(j)}_{\bfx}\right)_{\bfx\in\FF_2^m}$ and our task is to determine the parity sum $p^{(j)}_{\bfx}$ in terms of the $a^{(j)}_\bfx$'s. 
Now, in what follows, each parity sum in $\bfp^{(j)}$ is a linear combination $m+1$ information symbols and is determined by an $(m+1)$-tuple $\bfS_j$.  %Specifically, we do the following. do not consider m'=2

\begin{itemize}\itemindent=-10pt
\item Set $\bfS_0$ to be $(\vzero,\ldots, \vzero)$ and $S_1=\{\vzero,\bf \bfd_m,\bfe_2,\bfe_3,\ldots,\bfe_m \}$. 
\item For $j\in [2,m'-1]$, we set 
{\small
\begin{equation*}
    S_j=\{\vzero,
\bfe_{m-2j+3},\bfe_{m-2j+4},\ldots, \bfe_m,  \underbrace{\bfd_m}_{2j-th},\bfe_2,\ldots,\bfe_{m-2j+2} \},
\end{equation*}}
where $\bfd_m$ is the $2j$-th element of set $\bfS_j$. For example, if $j=2$, we have  $\bfS_2=\{\vzero, \bfe_{m-1},\bfe_{m},\bf d_m, \bfe_2,\ldots,e_{m-2}  \}$; if $j=3$, $\bfS_3=\{\vzero, \bfe_{m-3},\bfe_{m-2},\bfe_{m-1},\bfe_{m},\bf d_m, \bfe_2,\ldots,e_{m-4} \}$.

\item For the last parity columns,  $S_{m'}=\{\vzero,\bfd_2,\bfd_3,\ldots,\bfd_m,\bfd_1\}$.

In particular, if $m'=2$, we choose $S_1=\{\vzero,\bf \bfd_m,\bfe_2,\bfe_3,\ldots,\bfe_m \}$ and  $S_{2}=\{\vzero,\bfd_2,\bfd_3,\ldots,\bfd_m,\bfd_1\}$.

% For $j\in [2,m-2]$, we set 
% {\small
% \begin{equation*}
%     S_j=\{\vzero,
% \bfe_{m-2j+3},\bfe_{m-2j+4},\ldots, \bfe_m,  \underbrace{\bfd_m}_{(2j+1)-th},\bfe_2,\ldots,\bfe_{m-2j+2} \}.
% \end{equation*}}
% For the last two parity columns, if $2\mid m+1$,  $S_{m'-1}=\{\vzero,\ldots,\bfe_m,\bfd_m,\bfe_2,\bfe_3\}$. $S_{m'}=\{\vzero,\bfd_2,\bfd_3,\ldots,\bfd_m,\bfd_1\}$, if $2\nmid m+1$, $S_{m'-1}=\{\vzero,\bfe_3,\ldots,\bfe_m,\bfd_m,\bfe_2\}$, $S_{m'}=\{\vzero,\bfd_2,\bfd_3,\ldots,\bfd_m,\bfd_1\}$
% $S_1=\{\vzero,\bf \bfd_m,\bfe_2,\bfe_3,\ldots,\bfe_m \}$, $S_2=\{\vzero,\bfe_{m-1},\bfe_m,\bfd_m,\bfe_2,\ldots,\bfe_{m-2}\}$, $S_3=\{\vzero,\bfe_{m-3},\bfe_{m-2},\bfe_{m-1},\bfe_m,\bfd_m,\bfe_2,\ldots,\bfe_{m-4}\}$, $\ldots$.
% Here, if $2\mid m+1$,  $S_{m'-1}=\{\vzero,\ldots,\bfe_m,\bfd_m,\bfe_2,\bfe_3\}$. $S_m'=\{\vzero,\bfd_2,\bfd_3,\ldots,\bfd_m,\bfd_1\}$, if $2\nmid m+1$, $S_{m'-1}=\{\vzero,\bfe_3,\ldots,\bfe_m,\bfd_m,\bfe_2\}$, $S_m'=\{\vzero,\bfd_2,\bfd_3,\ldots,\bfd_m,\bfd_1\}$
\item Similar to the Construction A, we use the tuple $\bfS_j$ to generate a parity column $\bfp^{(j)}$. 
If $\bfS_j=(\bfv_0,\ldots, \bfv_m)$ and we consider row $\bfx$, we set $\bfu_i=\bfx+\bfv_i$ and define $p^{(j)}_{\bfx} = \sum^{m}_{i=0}\alpha^{(i,j)}_{\bfx}a^{(i)}_{\bfu_i}$ for some choice of coefficients $\alpha^{(i,j)}_{\bfx}\in\FF_q$ 
% \item For example, for the first parity column $\bfp^{(0)}$, we have that $\bfS_0=(\vzero,\ldots, \vzero)$.
% Then for row $\bfx$, the parity sum is $p^{(0)}_{\bfx}=\alpha^{(0,0)}_{\bfx}a^{(0)}_{\bfx}+\cdots+\alpha^{(m,0)}_{\bfx}a^{(m)}_{\bfx}$.
% In other words, $p^{(0)}_{\bfx}$ is a linear combination of the $(m+1)$ information symbols in the same row.
\end{itemize}

\vspace{1mm}
\noindent\textit{Repair of Information Node $s$ for $s\in [0,m]$}:

\begin{itemize}\itemindent=-10pt
\item If $s\in [m-1]$ satisfies  $2\nmid s$, set helper nodes to be $\cH = \{\bfa^{(i)}: i\in [0,m]\setminus\{s\}\}\cup\{\bfp^{(0)},\bfp^{(\frac{s+1}{2}  )}\}$. 
\begin{itemize}\itemindent=-20pt
    \item From $\bfa^{(i)}$, we read $\Read_{\bfa^{(i)}}=\{a^{(i)}_\bfx:\bfx\in\cU\}$.
    \item From $\bfp^{(i)}$ with $i\in\{0,\frac{s+1}{2}\}$, we read $\Read_{\bfp^{(i)}}=\{p^{(i)}_\bfx:\bfx\in\cU\}$.
\end{itemize}
\item If $s\in [m-1]$ satisfies  $2\mid s$, set helper nodes to be $\cH = \{\bfa^{(i)}: i\in [0,m]\setminus\{s\}\}\cup\{\bfp^{(0)},\bfp^{(\frac{s}{2}  )}\}$. 
\begin{itemize}\itemindent=-20pt
    \item From $\bfa^{(i)}$, we read $\Read_{\bfa^{(i)}}=\{a^{(i)}_\bfx:\bfx\in\cL'\}$.
    \item From $\bfp^{(i)}$ with $i\in\{0,\frac{s}{2}\}$, we read $\Read_{\bfp^{(i)}}=\{p^{(i)}_\bfx:\bfx\in\cL'\}$.
\end{itemize}
\item If $s=0$, set the helper nodes to be $\cH = \{\bfa^{(i)}: i\in [m]\}\cup\{\bfp^{(0)},\bfp^{(m')}\}$. 
\begin{itemize}\itemindent=-20pt
    \item From $\bfa^{(i)}$ with $i\in [m]$, we read $\Read_{\bfa^{(i)}}=\{a^{(i)}_\bfx:\bfx\in\cU\}$.
    \item From $\bfp^{(0)}$, we read $\Read_{\bfp^{(0)}}=\{p^{(0)}_\bfx:\bfx\in\cU\}$.
    \item From $\bfp^{(m')}$, we read $\Read_{\bfp^{(m)}}=\{p^{(m')}_\bfx:\bfx\in\cL\}$.
\end{itemize}

\item If $s=m$, set the helper nodes to be $\cH = \{\bfa^{(i)}: i\in [m]\}\cup\{\bfp^{(0)},\bfp^{(m')}\}$. 
\begin{itemize}\itemindent=-20pt
    \item From $\bfa^{(i)}$ with $i\in [m]$, we read $\Read_{\bfa^{(i)}}=\{a^{(i)}_\bfx:\bfx\in\cL'\}$.
    \item From $\bfp^{(i)}$ with $i\{0,m'\}$, we read $\Read_{\bfp^{(m')}}=\{p^{(0)}_\bfx:\bfx\in\cL'\}$.
\end{itemize}
\end{itemize}

The construction of 
systematic columns is similar to the Construction A. The difference between Construction A and Construction B is the construction of the parity columns.

\begin{theorem}\label{thm:zigzag2}
The repair scheme for Construction~B is correct. 
That is, for $s\in [0,m]$, we are able to determine $\bfa^{(s)}$ from all reads.
Moreover, the repair scheme has zero skip cost and optimal rebuilding ratio $1/2$. Additionally, the code rate of MDS array codes from Construction~B is $\frac{m+1}{m+2+\lceil \frac{m+1}{2}\rceil}$, which is approximately $2/3$.
\end{theorem}

\begin{proof}
% The proof for $s\in \{0,1,\ldots,m \}$ is similar to the proof of Theorem~\ref{thm:zigzag1}.
We will give the proof of the case $s\in[2,m-1]$, the other cases of $s\in\{0,m\} $can be proved similarly. The proof of case $s\in[1,m-1]$ will be divided into two subcases.

Case(1):
For the case $2\nmid s$, we read the symbols from  $\bfp^{(0)}$ and $\bfp^{\frac{(s+1)}{2})}$ are
\begin{align*}
\Read_{\bfp^{(0)}}   &= \left\{ \sum_{i=0}^m \alpha_\bfx^{(i,j)} a^{(i)}_{\bfu_i}: \bfx\in \cU\right\},\\
\Read_{\bfp^{(s+1)/2}} &= \left\{ \sum_{i=0}^m \alpha_\bfx^{(i,j)} a^{(i)}_{\bfu_i}: \bfx\in \cU\right\}.    
\end{align*}
Here, let $\cI_1$ and $\cI_2$ be the information symbols involved in the reads $\Read_{\bfp^{(0)}}$ and $\Read_{\bfp^{(s+1)/2}}$, respectively.

We observe that 
\begin{align}\label{eq:1}
\bfS_{0} &= (\vzero,\vzero,\ldots,{\color{blue}{\vzero}},{\color{red}{\vzero}},\ldots, \vzero),\\ \label{eq:2}
\bfS_{\frac{s+1}{2}} & =\{\vzero,
\bfe_{m-s+2},\bfe_{m-s+3},\ldots, \bfe_m, {\color{blue}\bfd_m},{\color{red}\bfe_2},\ldots,\bfe_{m-s+1} \}. 
\end{align}
when $s=1$, 
 $\bfS_0=(\vzero, {\color{blue}\vzero},{\color{red}\vzero}, \ldots, \vzero)$ and $S_1=\{\vzero,\bf {\color{blue}\bfd_m},{\color{red}\bfe_2},\bfe_3,\ldots,\bfe_m \}$, where the highlighted positions in blue are the $s$-th entries. Hence, with a similar discussion as Theorem~\ref{thm:zigzag1}, we have the following.
\begin{itemize} 
\setlength{\leftmargin}{0pt}
\item {\em $\cI_1\cup \cI_2$ contains all information symbols $a^{(s)}_\bfx$ with $\bfx\in\FF_2^m$}. 

\item {\em If $i\in\{0,\ldots,k\} 
\setminus s$, then $\cI_1\cup \cI_2$ contains exactly all information symbols $a^{(i)}_\bfx$ with $\bfx\in\cU$}. Indeed, in both $\bfS_{0}$ and $\bfS_{(s+1)/2}$, the $i$-th entries belong to $\cU$. Hence, we have that $a^{(i)}_{\bfu_i}$ always belong to $\cU$.
\end{itemize}

Case(2): For the case $2\nmid s$, we read the symbols from  $\bfp^{(0)}$ and $\bfp^{\frac{(s+1)}{2})}$ are
\begin{align*}
\Read_{\bfp^{(0)}}   &= \left\{ \sum_{i=0}^m \alpha_\bfx^{(i,j)} a^{(i)}_{\bfu_i}: \bfx\in \cL'\right\},\\
\Read_{\bfp^{(s+1)/2}} &= \left\{ \sum_{i=0}^m \alpha_\bfx^{(i,j)} a^{(i)}_{\bfu_i}: \bfx\in \cL'\right\}.    
\end{align*}

From equations~\eqref{eq:1} and~\eqref{eq:2}, the
highlighted positions in red are the s-th entries.
Hence, with a similar discussion as Theorem 4, we have the following.

\begin{itemize} 
\setlength{\leftmargin}{0pt}
\item  $\cI_1\cup \cI_2$ contains all information symbols $a^{(s)}_\bfx$ with $\bfx\in\FF_2^m$. Here, for a $\bfx\in \cU'$, if $\bfy\in \cU'$, we have $\bfx+\bfy\in \cU'$; if  $\bfy\in \cL'$, we have $\bfx+\bfy\in \cL'$. Similarly, for a $\bfx\in \cL'$, if $\bfy\in \cL'$, we have $\bfx+\bfy\in \cU'$. Obviously, $\{\bfe_{m-s+2},\ldots,\bfe_{m},\bfd_m,\bfe_3,\ldots,\bfe_{m-s+1}\} \subseteq \cU'$, and $\bfe_2\in\cL'$.
\item  If $i\in\{0,\ldots,k\} 
\setminus s$, then $\cI_1\cup \cI_2$ contains exactly all information symbols $a^{(i)}_\bfx$ with $\bfx\in\cL'$. Indeed, in both $\bfS_{0}$ and $\bfS_{(s+1)/2}$, the $i+1$-th entries belong to $\cL'$. Hence, we have that $a^{(i)}_{\bfu_i}$ always belong to $\cL'$.
\end{itemize}

The cases $s\in\{0,m\}$ can be proved in a similar way. Therefore, using the reads $\Read_{\bfa^{(i)}}~(i\ne s)$, we eliminate all information symbols $a^{(i)}_\bfx$ with $i\ne s$ from the corresponding check sums and recover $\bfa^{(s)}$.

It is clear that the skip cost of this repair scheme is zero, as we consistently download symbols from all consecutive rows in either $\cU$ or $\cL'$. Consequently, the rebuilding ratio stands at $\frac{m+1}{m+2+\lceil \frac{m+1}{2}\rceil}$. 
As $m$ increases, the rebuilding ratio approaches $\frac{2}{3}$. Furthermore, 
this is optimal as the number of helper nodes is  $m+2= k+1$.

\end{proof}

% $$-----------$$
% \begin{remark}

% \end{remark}

% $$-----------$$
\begin{figure*}[!t]
\centering
   % \vspace{2mm}
    
%\footnotesize
 %\noindent(a) $(8\times 7,4)$-array code code with skip cost zero.
%\vspace{1mm}
{\centering \scriptsize
\setlength\tabcolsep{1pt} 
    \begin{tabular}{|c||c|c|c|c|c|c|c|}
    \hline
    & $\bfa^{(0)}$ & $\bfa^{(1)}$ & $\bfa^{(2)}$ & $\bfa^{(3)}$ &
      $\bfp^{(0)}$ & $\bfp^{(1)}$ & $\bfp^{(2)}$ \\
    & $\spadesuit$  & $\heartsuit$ &  $\clubsuit$ & $\diamondsuit$   & 
    $\bfS_0=(\vzero,\vzero,\vzero,\vzero)$ & $\bfS_1=(\vzero,\bfd_3,\bfe_2,\bfe_3)$ & $\bfS_2=(\vzero,\bfd_2,\bfd_3,\bfd_1)$ \\ \hline 
    $000$ & $000_\spadesuit$ & $000_\heartsuit$  & {\color{red}$000_\clubsuit$} & $000_\diamondsuit$ & 
$000_{ \spadesuit }\boxplus 000_{\heartsuit }\boxplus 000_{\clubsuit} \boxplus 000_{\diamondsuit }$  & 
   $000_{ \spadesuit }\boxplus 111_{\heartsuit }\boxplus 010_{\clubsuit}\boxplus 001_{\diamondsuit} $ &
    $000_{ \spadesuit }\boxplus 110_{\heartsuit }\boxplus 111_{\clubsuit}\boxplus 100_{\diamondsuit} $  \\ 
    
    $001$ & $001_\spadesuit$ & $001_\heartsuit$  & {\color{red}$001_\clubsuit$ }&
    $001_\diamondsuit$ & 
$001_{ \spadesuit }\boxplus 001_{\heartsuit}\boxplus 001_{\clubsuit} \boxplus 001_{\diamondsuit}  $ &
$001_{ \spadesuit }\boxplus 110_{\heartsuit}\boxplus 011_{\clubsuit}\boxplus 000_{\diamondsuit}  $ &
    $001_{ \spadesuit }\boxplus 111_{\heartsuit }\boxplus 110_{\clubsuit}\boxplus 101_{\diamondsuit} $  \\ 
    
    $010$ & {\color{blue}$010_\spadesuit$} & {\color{blue}$010_\heartsuit$}  & {\color{red}$010_\clubsuit$} &
    {\color{blue}$010_\diamondsuit$} & 
    {\color{blue} $010_{\spadesuit }\boxplus 010_{\heartsuit}\boxplus 010_{\clubsuit}\boxplus 010_{\diamondsuit}  $} &
   {\color{blue} $010_{\spadesuit }\boxplus 101_{\heartsuit }\boxplus 000_{\clubsuit}\boxplus 011_{\diamondsuit}$} &
    $010_{\spadesuit}\boxplus 100_{\heartsuit}\boxplus 101_{\clubsuit}\boxplus 110_{\diamondsuit}$  \\

    $011$ & {\color{blue}$011_\spadesuit$} & {\color{blue}$011_\heartsuit$ } & {\color{red} $011_\clubsuit$} & {\color{blue}$011_\diamondsuit$} & 
    {\color{blue}$011_{\spadesuit }\boxplus 011_{\heartsuit }\boxplus 011_{\clubsuit}\boxplus 011_{\diamondsuit} $ }& 
   {\color{blue}$011_{\spadesuit }\boxplus 100_{\heartsuit }\boxplus 001_{\clubsuit}\boxplus 010_{\diamondsuit} $ }& 
    $011_{\spadesuit }\boxplus 101_{\heartsuit }\boxplus 100_{\clubsuit}\boxplus 111_{\diamondsuit} $  \\ 

    $100$ &  {\color{blue}$100_\spadesuit$ }&  {\color{blue}$100_\heartsuit$  }&  {\color{red} $100_\clubsuit$} & {\color{blue}$100_\diamondsuit$} & 
    {\color{blue} $100_{\spadesuit }\boxplus 100_{\heartsuit}\boxplus 100_{\clubsuit}\boxplus 100_{\diamondsuit}$} & 
 {\color{blue}  $100_{\spadesuit }\boxplus 011_{\heartsuit }\boxplus 110_{\clubsuit}\boxplus 101_{\diamondsuit}$ }& 
    $100_{\spadesuit}\boxplus 010_{\heartsuit}\boxplus 011_{\clubsuit}\boxplus 000_{\diamondsuit}$  \\ 
    
      $101$ & {\color{blue}$001_\spadesuit$} &{\color{blue} $101_\heartsuit$}  & {\color{red}$101_\clubsuit$} &
    {\color{blue}$101_\diamondsuit$} & 
    {\color{blue}$101_{\spadesuit }\boxplus 101_{\heartsuit}\boxplus 101_{\clubsuit} \boxplus 101_{\diamondsuit} $} &
    {\color{blue}$101_{\spadesuit }\boxplus 010_{\heartsuit}\boxplus 111_{\clubsuit}\boxplus 100_{\diamondsuit}$} &
    $101_{\spadesuit}\boxplus 011_{\heartsuit}\boxplus 010_{\clubsuit}\boxplus 001_{\diamondsuit} $  \\ 
    
     $110$ & $110_\spadesuit$ & $110_\heartsuit$ & {\color{red}$110_\clubsuit$} &
    $110_\diamondsuit$ & 
    $110_{\spadesuit}\boxplus 110_{\heartsuit}\boxplus 110_{\clubsuit}\boxplus 110_{\diamondsuit}  $ &
    $110_{\spadesuit}\boxplus 001_{\heartsuit}\boxplus 100_{\clubsuit}\boxplus 111_{\diamondsuit} $ &
    $110_{\spadesuit}\boxplus 000_{\heartsuit}\boxplus 001_{\clubsuit}\boxplus 010_{\diamondsuit} $  \\

    $111$ & $111_\spadesuit$ & $111_\heartsuit$  & {\color{red}$111_\clubsuit$} & $111_\diamondsuit$ & 
    $111_{\spadesuit}\boxplus 111_{\heartsuit}\boxplus 111_{\clubsuit}\boxplus 111_{\diamondsuit} $ & 
  $111_{\spadesuit}\boxplus 000_{\heartsuit}\boxplus 101_{\clubsuit}\boxplus 110_{\diamondsuit} $ & 
    $111_{\spadesuit}\boxplus 001_{\heartsuit}\boxplus 000_{\clubsuit}\boxplus 011_{\diamondsuit} $  \\ \hline
    \end{tabular}}
       \caption{
    Example of a {$(8\times 7,4)$}-MDS array code described by Construction~B with $m=3$ (see Section~\ref{sec:zigzag}). Suppose information node $\bfa^{(2)}$ (highlighted in {\color{red}red}) fails. We contact nodes $\bfa^{(0)}$, $\bfa^{(1)}$, $\bfa^{(3)}$, $\bfp^{(0)}$ and $\bfp^{(1)}$ and read the contents in {\color{blue}blue}. Here, the skip cost is zero.
    % Note that we use $\bfx_{i}$ to represent the information symbol $a_\bfx^{(i)}$, while the `sum' $\bf x_i\boxplus \bfy_j\boxplus \bfz_k$ indicates that the corresponding codesymbol is a linear combination of $a_\bfx^{(i)}$, $a_\bfy^{(j)}$, and $a_\bfz^{(k)}$. %\hm{Use plus.}
    }\label{fig:zigzag2}

\end{figure*}

\begin{example}
    Let $m=3$. Hence, $M=8$ and $N=7$. We consider MDS array code with four systematic information columns three parity columns is as follows. The placement of codesymbols is presented as Fig {\color{red}\ref{fig:zigzag2}}. Similar to the Example 1, we index the information nodes with use the symbols  $\{\spadesuit,\heartsuit,\clubsuit,\diamondsuit\}$ in place of $\{0,1,2,3\}$.  Additionally, instead of writing $a^{(i)}_{\boldsymbol x}$, we write it as $\boldsymbol x_{i}$. In other words, $000_{\spadesuit}$ represents the information symbol $a^{(\spadesuit)}_{000}$. For the parity symbols, we use ${\boldsymbol x}_i \boxplus {\boldsymbol y}_j \boxplus {\boldsymbol z}_k \boxplus {\boldsymbol w}_t$  to represent a linear combination of $a_{\boldsymbol x^{(i)}}$, $a_{\boldsymbol y^{(j)}}$, $a_{\boldsymbol z^{(k)}}$, and $a_{\boldsymbol w^{(t)}}$ .
     
Assume that the information node corresponding to $\clubsuit$ (or $s=1$)   in Fig.~\ref{fig:zigzag2}(a) has failed. From our scheme, the helper information nodes are $\spadesuit$, $\heartsuit$, and $\diamondsuit$, while the helper parity nodes are $\bfp^{(0)}$ and $\bfp^{(1)}$. Specifically, the reads from the parity nodes contain the following information symbols.

    \begin{itemize}
    \item From $\boldsymbol{p}^{(0)}$, we obtain
\begin{align} \nonumber
\mathcal{I}_1& \triangleq &
    \{010_{ \spadesuit },010_{\heartsuit }, {\color{red}010_{\clubsuit}}, 010_{\diamondsuit}\}   \cup \{011_\spadesuit,011_\heartsuit,{\color{red}011_\clubsuit},011_\diamondsuit 
    \} 
    \cup \{100_\spadesuit, 100_\heartsuit,{\color{red}100_\clubsuit},100_\diamondsuit  \}  \cup \{ 101_\spadesuit,101_\heartsuit,{\color{red}101_\clubsuit},101_\diamondsuit\}
        \end{align}

\item From $\boldsymbol p^{(1)}$, we obtain
 \begin{align} \nonumber
 \mathcal {I}_2 &\triangleq& 
   \{010_{\spadesuit }, 101_{\heartsuit},{\color{red}000_\clubsuit},011_\diamondsuit \}\cup
    \{011_{\spadesuit }, 100_{\heartsuit},{\color{red}001_\clubsuit}, 010_{\diamondsuit}\}\
     \cup  \{ 100_{\spadesuit},011_{\heartsuit},{\color{red}110_{\clubsuit}},101_{\diamondsuit}\}  \cup \{101_{\spadesuit},010_{\heartsuit},{\color{red}111_{\clubsuit}},100_{\diamondsuit} \}
        \end{align} 
\end{itemize}
    Indeed, we check that 
\begin{equation*}
\mathcal {I}_1\cup\mathcal {I}_2 = \{\boldsymbol{x}_{\spadesuit}: \boldsymbol{x}\in \mathcal L'\}\cup
\{\boldsymbol{x}_{\heartsuit}: \boldsymbol{x} \in \mathcal L'\}
\cup
\{\boldsymbol{x}_{\clubsuit}: \boldsymbol{x}\in\mathbb{F}_2^3\}\cup
\{\boldsymbol{x}_{\diamondsuit}:x\in \mathcal{L'} \},
\end{equation*}
    which corroborates with the claims in the proof of Theorem~\ref{thm:zigzag2}.
    Since we download the information symbols $\{\boldsymbol{x}_{\spadesuit}: \boldsymbol{x}\in\mathcal{L'}\}$ , $\{\boldsymbol{x}_{\heartsuit} : \boldsymbol{x}\in\mathcal{L'}\}$ 
 and $\{\boldsymbol{x}_{\diamondsuit} : \boldsymbol{x}\in\mathcal{L'}\}$ 
 from $\boldsymbol{a}^{(0)}$, $\boldsymbol{a}^{(1)}$, and $\boldsymbol{a}^{(3)}$ respectively, we recover all values in $\{\boldsymbol{x}_{\clubsuit}: \boldsymbol{x}\in\mathbb{F}_2^3\}$, i.e., $\boldsymbol{a}^{(2)}$.
    Since we download contents from rows $\mathcal{L'}$, the skip cost is zero.

 Similarly, if information node  $\heartsuit$ ($\bfa^{(1)}$) is erased. From Fig.~\ref{fig:zigzag2}, the helper information nodes are $\spadesuit$, $\clubsuit$, and $\diamondsuit$, while the helper parity nodes also are $\bfp^{(0)}$ and $\bfp^{(1)}$. The reads from the parity nodes contain the following information symbols.
   \begin{itemize}
    \item From $\boldsymbol{p}^{(0)}$, we obtain
 {\small\begin{align} \nonumber
\mathcal{I}_1& \triangleq &
    \{000_{ \spadesuit },{\color{red}000_{\heartsuit }}, 000_{\clubsuit}, 000_{\diamondsuit}\}  \cup \{001_\spadesuit,{\color{red}001_\heartsuit},001_\clubsuit,001_\diamondsuit 
    \} 
    \cup \{010_\spadesuit, {\color{red}010_\heartsuit},010_\clubsuit,010_\diamondsuit  \}  \cup \{ 011_\spadesuit,{\color{red}011_\heartsuit},011_\clubsuit,011_\diamondsuit\}
        \end{align}}

\item From $\boldsymbol p^{(1)}$, we obtain
 {\small\begin{align} \nonumber
 \mathcal {I}_2 &\triangleq& 
   \{000_{\spadesuit }, {\color{red}111_{\heartsuit}},010_\clubsuit,001_\diamondsuit \}\cup
    \{001_{\spadesuit }, {\color{red}110_{\heartsuit}},011_\clubsuit, 000_{\diamondsuit}\}\ 
     \cup  \{ 010_{\spadesuit}, {\color{red}101_{\heartsuit}},000_{\clubsuit},011_{\diamondsuit}\}  \cup \{011_{\spadesuit},{\color{red}100_{\heartsuit}},001_{\clubsuit},010_{\diamondsuit} \}
        \end{align}}
\end{itemize}
    Indeed, we check that 
{\small
\begin{equation*}
\mathcal {I}_1\cup\mathcal {I}_2 = \{\boldsymbol{x}_{\spadesuit}: \boldsymbol{x}\in \mathcal U\}\cup
\{\boldsymbol{x}_{\heartsuit}: \boldsymbol{x} \in \mathbb{F}_2^3\}
\cup
\{\boldsymbol{x}_{\clubsuit}: \boldsymbol{x}\in\mathcal{U}\}\cup
\{\boldsymbol{x}_{\diamondsuit}:x\in \mathcal{U} \}.
\end{equation*}}
\noindent Combining with nodes $\boldsymbol{p}^{(0)}$ and $\boldsymbol{p}^{(1)}$, it is clear that we can recover all values in $\{\boldsymbol{x}_{\heartsuit}: \boldsymbol{x}\in\mathbb{F}_2^3\}$ by downloading the information symbols $\{\boldsymbol{x}_{\spadesuit}: \boldsymbol{x}\in\mathcal{L'}\}$, $\{\boldsymbol{x}_{\clubsuit} : \boldsymbol{x}\in\mathcal{L'}\}$, and $\{\boldsymbol{x}_{\diamondsuit} : \boldsymbol{x}\in\mathcal{L'}\}$ from $\boldsymbol{a}^{(0)}$, $\boldsymbol{a}^{(2)}$, and $\boldsymbol{a}^{(3)}$.

Additionally, if node $\boldsymbol{a}^{(0)}$ is failed,  we repair it by reading $ \{p_\bfx: \bfx\in \cU \}$ and $ \{p_\bfx: x\in \cL \}$  from parity nodes $\boldsymbol{p}^{(0)}$, $\boldsymbol{p}^{(2)}$ and $\{\boldsymbol{x}_{\heartsuit}: \boldsymbol{x}\in\mathcal{U}\}$, $\{\boldsymbol{x}_{\clubsuit} : \boldsymbol{x}\in\mathcal{U}\}$, and $\{\boldsymbol{x}_{\diamondsuit} : \boldsymbol{x}\in\mathcal{U}\}$ from $\boldsymbol{a}^{(0)}$, $\boldsymbol{a}^{(2)}$, and $\boldsymbol{a}^{(3)}$. If node $\boldsymbol{a}^{(3)}$ is failed,  we repair it by reading $ \{p_\bfx: \bfx\in \cL' \}$ and $ \{p_\bfx: x\in \cL' \}$  from parity nodes $\boldsymbol{p}^{(0)}$, $\boldsymbol{p}^{(2)}$ and $\{\boldsymbol{x}_{\heartsuit}: \boldsymbol{x}\in\mathcal{L'}\}$, $\{\boldsymbol{x}_{\clubsuit} : \boldsymbol{x}\in\mathcal{L'}\}$, and $\{\boldsymbol{x}_{\diamondsuit} : \boldsymbol{x}\in\mathcal{L'}\}$ from $\boldsymbol{a}^{(0)}$, $\boldsymbol{a}^{(2)}$, and $\boldsymbol{a}^{(3)}$. Obviously, our repair scheme for each information node have zero skip cost. \qedsymbol
\end{example}

% $$---------$$

% {\color{red}It remains to choose the coefficients $\alpha^{(i,j)}_\bfx$ so as to achieve the MDS property.
% As the proof is technical and similar to that in~\cite{Tamo2013}, we defer the proof of the following proposition to Appendix~\ref{app:mds}. }%\hm{Wenqin, can you write the proof?}
%\hm{To write example. To write remarks about differences with Tamo-Wang-Bruck. To write remarks about optimal repair. To prove MDS!!}
% $$----------------$$

%$$--------------$$

%Note that Lower sub-packetization levels offer the advantage of facilitating system implementation while also potentially dispersing the burden of supplying repair information for a failed node across numerous nodes. Consequently, large-scale distributed storage systems necessitate MDS array codes with small sub-packetization levels.
%Hence, we consider the MDS array code with any information columns $k$ for the fixed $M$. In other words, we can give the  construction of MDS array codes with small sub-packetization and zero skip cost while ensuring optimal updates.

Now, lower sub-packetization levels simplify implementation while  distributing the responsibility of providing repair information for a failed node across multiple nodes. Consequently, large-scale distributed storage systems require MDS array codes with small sub-packetization levels.
Therefore, in the next construction, we focus on MDS array codes with any information columns $k$ for a fixed subpacketization level $M$. 
In other words, we construct MDS array codes with small sub-packetization and zero skip cost.

 We retain the previous notation $\mathcal{U}$, $\mathcal{U}'$, $\mathcal{L}$, and $\mathcal{L}'$ in Construction B for the sake of simplicity. Next, we present a generalization of Construction B, wherein the number of information nodes $k$ does not depend on the sub-packetization $M$.

\vspace{1mm}
\noindent{\bf Construction C.} 
Let $m\geq2$ and $k\geq2$ be positive integers. We present an $(M\times N,k)$-MDS array code with $M=2^m$ packets and $N=k+ \left\lceil\frac{k}{2} \right\rceil+1$ nodes. Let $m'= \left\lceil\frac{k}{2} \right\rceil$. 
Of these $N$ nodes, $k$ of them are {\em systematic} nodes $\bfa^{(0)},\bfa^{(1)}\ldots, \bfa^{(k-1)}$, while the remaining $m'$ are parity nodes $\bfp^{(0)}, \bfp^{(1)},\ldots, \bfp^{(m')}$.

\vspace{1mm}
\noindent\textit{Contents of systematic columns}: For $j\in[0,k-1]$, we simply set $\bfa^{(j)}=\left(a^{(j)}_{\bfx}\right)_{\bfx\in\FF_2^m}$.
\vspace{1mm}

\noindent\textit{Contents of parity columns}: For $j\in[0,m']$, we set $\bfp^{(j)}=\left(p^{(j)}_{\bfx}\right)_{\bfx\in\FF_2^m}$. 
Now, in what follows, each parity sum in $\bfp^{(j)}$ is a linear combination $k$ information symbols and is determined by an $k$-tuple $\bfS_j$ for $j\in[m']$.  %Specifically, we do the following. do not consider m'=2

\begin{itemize}\itemindent=-10pt

\item Set $\bfS_0$ to be $(\vzero,\ldots, \vzero)$.
\item For $j\in [1,m'-1]$, we set 
{\small
\begin{equation*}
    S_j=\{\vzero,
\vzero,\vzero,\ldots, \vzero,  \underbrace{\bfd_m}_{2j-th},\bfe_2,\vzero,\ldots,\vzero \},
\end{equation*}}
where $\bfd_m$ is the $2j$-th element of set $\bfS_j$. For example, if $j=2$, we have $\bfS_2=\{\vzero, \vzero,\vzero,\bf d_m, \bfe_2, \ldots,\vzero \}$; if $j=3$, $\bfS_3=\{\vzero, \vzero,\vzero,\vzero,\vzero, \bf d_m, \bfe_2,\vzero,\ldots,\vzero \}$.

\item For the last parity columns,  $S_{m'}=\{\vzero,\bfd_2,\bfd_2,\ldots,\bfd_2,\bfd_1\}$.

In particular, if $m'=2$, we choose $S_1=\{\vzero,\bf \bfd_m,\bfe_2,\vzero\}$ and $S_{2}=\{\vzero,\bfd_2,\bfd_2,\bfd_1\}$.

\item Similar to the Construction A, we use the tuple $\bfS_j$ to generate a parity column $\bfp^{(j)}$. 
If $\bfS_j=(\bfv_0,\ldots, \bfv_{k-1})$ and we consider row $\bfx$, we set $\bfu_i=\bfx+\bfv_i$ and define $p^{(j)}_{\bfx} = \sum^{m}_{i=0}\alpha^{(i,j)}_{\bfx}a^{(i)}_{\bfu_i}$ for some choice of coefficients $\alpha^{(i,j)}_{\bfx}\in\FF_q$ 
\end{itemize}

\vspace{1mm}
\noindent\textit{Repair of Information Node $s$ for $s\in [0,k-1]$}:

\begin{itemize}\itemindent=-10pt
\item If $s\in [k-2]$ satisfies  $2\nmid s$, set helper nodes to be $\cH = \{\bfa^{(i)}: i\in [0,k-1]\setminus\{s\}\}\cup\{\bfp^{(0)},\bfp^{(\frac{s+1}{2}  )}\}$. 
\begin{itemize}\itemindent=-20pt
    \item From $\bfa^{(i)}$, we read $\Read_{\bfa^{(i)}}=\{a^{(i)}_\bfx:\bfx\in\cU\}$.
    \item From $\bfp^{(i)}$ with $i\in\{0,\frac{s+1}{2}\}$, we read $\Read_{\bfp^{(i)}}=\{p^{(i)}_\bfx:\bfx\in\cU\}$.
\end{itemize}
\item If $s\in [k-2]$ satisfies  $2\mid s$, set helper nodes to be $\cH = \{\bfa^{(i)}: i\in [0,k-1]\setminus\{s\}\}\cup\{\bfp^{(0)},\bfp^{(\frac{s}{2}  )}\}$. 
\begin{itemize}\itemindent=-20pt
    \item From $\bfa^{(i)}$, we read $\Read_{\bfa^{(i)}}=\{a^{(i)}_\bfx:\bfx\in\cL'\}$.
    \item From $\bfp^{(i)}$ with $i\in\{0,\frac{s}{2}\}$, we read $\Read_{\bfp^{(i)}}=\{p^{(i)}_\bfx:\bfx\in\cL'\}$.
\end{itemize}
\item If $s=0$, set the helper nodes to be $\cH = \{\bfa^{(i)}: i\in [k-1]\}\cup\{\bfp^{(0)},\bfp^{(m')}\}$. 
\begin{itemize}\itemindent=-20pt
    \item From $\bfa^{(i)}$ with $i\in [k-1]$, we read $\Read_{\bfa^{(i)}}=\{a^{(i)}_\bfx:\bfx\in\cU\}$.
    \item From $\bfp^{(0)}$, we read $\Read_{\bfp^{(0)}}=\{p^{(0)}_\bfx:\bfx\in\cU\}$.
    \item From $\bfp^{(m')}$, we read $\Read_{\bfp^{(m')}}=\{p^{(m')}_\bfx:\bfx\in\cL\}$.
\end{itemize}

\item If $s=m$, set the helper nodes to be $\cH = \{\bfa^{(i)}: i\in [k-1]\}\cup\{\bfp^{(0)},\bfp^{(m')}\}$. 
\begin{itemize}\itemindent=-20pt
    \item From $\bfa^{(i)}$ with $i\in [k-1]$, we read $\Read_{\bfa^{(i)}}=\{a^{(i)}_\bfx:\bfx\in\cL'\}$.
    \item From $\bfp^{(i)}$ with $i\{0,m'\}$, we read $\Read_{\bfp^{(m')}}=\{p^{(0)}_\bfx:\bfx\in\cL'\}$.
\end{itemize}
\end{itemize}

\begin{theorem}\label{thm:zigzag3}
The repair scheme for Construction~C is correct. 
That is, for $s\in [0,k-1]$, we are able to determine $\bfa^{(s)}$ from all reads.
Moreover, the repair scheme has zero skip cost and optimal rebuilding ratio $1/2$. Additionally, the code rate of MDS array codes from Construction~C is $\frac{k}{k+1+\lceil \frac{k}{2}\rceil}$, which is approximately $2/3$.
\end{theorem}

\begin{proof}
    The proof is similar to the proof of Theorem~\ref{thm:zigzag2} so the proof is omitted.
\end{proof}

\begin{figure*}[!t]
%\centering
 %   \vspace{2mm}
    
%\footnotesize
 %\noindent $(4\times 10,6)$-array code code with skip cost zero.
%\vspace{1mm}{
\centering \scriptsize
\setlength\tabcolsep{1pt} 
    \begin{tabular}{|c|c|c|c|c|c|c|c|c|c|c|}
    \hline
\multirow{2}{*}{$\bfa^{(0)}$}  & \multirow{2}{*}{$\bfa^{(1)}$} & \multirow{2}{*}{$\bfa^{(2)}$} & \multirow{2}{*}{$\bfa^{(3)}$}&\multirow{2}{*}{$\bfa^{(4)}$} & \multirow{2}{*}{$\bfa^{(5)}$} &
      $\bfp^{(0)}$ & $\bfp^{(1)}$ & $\bfp^{(2)}$ & $\bfp^{(4)}$ \\
  & & & & & &  $\bfS_0=(\vzero,\vzero,\vzero,\vzero,\vzero,\vzero)$ & $\bfS_1=(\vzero,\bfd_2,\bfe_2,\vzero,\vzero,\vzero)$ & $\bfS_2=(\vzero,\vzero,\vzero,\bfd_2,\bfe_2,\vzero)$& $\bfS_2=(\vzero,\bfd_2,\bfd_2,\bfd_2,\bfd_2,\bfd_1)$ \\ \hline 
  
  $00_0$ & $00_1$  & {\color{red}$00_2$} & $00_3$ & $00_4$ & $00_5$   &
$00_{0}\boxplus\cdots\boxplus 00_{5}$ & 
$00_{0}\boxplus 11_{1 }\boxplus 01_{2} \boxplus 00_{3 }\boxplus 00_{4} \boxplus 00_{5 }$  & 
$00_{0}\boxplus 00_{1 }\boxplus 00_{2} \boxplus 11_{3 }\boxplus 01_{4} \boxplus 00_{5 }$  & 
$00_{0}\boxplus 11_{1 }\boxplus 11_{2} \boxplus 11_{3 }\boxplus 11_{4} \boxplus 10_{5 }$  \\ 
    
{\color{blue} $01_0$} & {\color{blue}$01_1$}  & {\color{red}$01_2$} & {\color{blue}$01_3$} & {\color{blue}$01_4$} & {\color{blue}$01_5$}   & {\color{blue}$01_{0}\boxplus\cdots\boxplus 01_{5}$}   
  & 
{\color{blue}$01_{0}\boxplus 10_{1 }\boxplus 00_{2} \boxplus 01_{3 }\boxplus 01_{4} \boxplus 01_{5 }$ } & 
$01_{0}\boxplus 01_{1 }\boxplus 01_{2} \boxplus 10_{3 }\boxplus 00_{4} \boxplus 01_{5 }$  & 
$01_{0}\boxplus 10_{1 }\boxplus 10_{2} \boxplus 10_{3 }\boxplus 10_{4} \boxplus 11_{5 }$   \\

 {\color{blue}$10_0$} & {\color{blue}$10_1$}  & {\color{red}$10_2$} & {\color{blue}$10_3$} & {\color{blue}$10_4$} & {\color{blue}$10_5$}   &{\color{blue} $10_{0}\boxplus\cdots\boxplus 10_{5}$   }
 & 
{\color{blue}$10_{0}\boxplus 01_{1 }\boxplus 11_{2} \boxplus 10_{3 }\boxplus 10_{4} \boxplus 10_{5 }$ } & 
$10_{0}\boxplus 10_{1 }\boxplus 10_{2} \boxplus 01_{3 }\boxplus 11_{4} \boxplus 10_{5 }$  & 
$10_{0}\boxplus 01_{1 }\boxplus 01_{2} \boxplus 01_{3 }\boxplus 01_{4} \boxplus 00_{5 }$   \\ 
    
$11_0$ & $11_1$  & {\color{red}$11_2$} & $11_3$ & $11_4$ & $11_5$   &
$11_{0}\boxplus\cdots\boxplus 11_{5}$     & 
$11_{0}\boxplus 00_{1 }\boxplus 10_{2} \boxplus 11_{3 }\boxplus 11_{4} \boxplus 11_{5 }$  & 
$11_{0}\boxplus 11_{1 }\boxplus 11_{2} \boxplus 00_{3 }\boxplus 10_{4} \boxplus 11_{5 }$  & 
$11_{0}\boxplus 00_{1 }\boxplus 00_{2} \boxplus 00_{3 }\boxplus 00_{4} \boxplus 01_{5 }$   
 \\ \hline
    \end{tabular}
\caption{Example of a {$(4\times 10,6)$}-MDS array code described by Construction~C with $m=2$ (see Section~\ref{sec:zigzag}). Suppose information node $\bfa^{(2)}$ (highlighted in {\color{red}red}) fails. We contact nodes $\bfa^{(i)}$ for $i \in\{0,1,3,4,5\}$, $\bfp^{(0)}$ and $\bfp^{(1)}$ and read the contents in {\color{blue}blue}. Here, the skip cost is zero.
    }\label{fig:zigzag3}
\end{figure*}

\begin{example}
    Let $m=2$ and set $k=6$. Hence, $M=4$, $N=10$. We consider an MDS array code with six systematic information columns three parity columns.
    The placement of codesymbols is presented as Fig~\ref{fig:zigzag3}. 
Similarly to Example 1, we index the information nodes using the symbols $a^{(i)}_{\boldsymbol x}$ for $i\in\{0,\ldots,5\}$.
    Furthermore, instead of writing $a^{(i)}_{\boldsymbol x}$, we write it as $\boldsymbol x_{i}$. For example, $00_{0}$ represents the information symbol in information nodes $a^{(0)}$. 
       
Assume that the information node $\bfa^{(2)}$ in Fig.~\ref{fig:zigzag3} has failed. In our scheme, the helper information nodes are $\bfa^{(i)}$ for $i\in\{0,1,3,4,5 \}$ while the helper parity nodes are $\bfp^{(0)}$ and $\bfp^{(1)}$. Specifically, the reads from the parity nodes contain the following information symbols.
    \begin{itemize}
    \item From $\boldsymbol{p}^{(0)}$, we obtain
\begin{align} \nonumber
\mathcal{I}_1& \triangleq 
\{01_0,01_1,{\color{red}01_2},01_3,01_4,01_5\} \cup  \{10_{0},10_1,{\color{red}01_2},10_3,10_4,10_5\} 
        \end{align}

\item From $\boldsymbol p^{(1)}$, we obtain
 \begin{align} \nonumber
 \mathcal {I}_2 &\triangleq
   \{01_0, 10_1, {\color{red}00_2},01_3,01_4,01_5\}\cup
   \{ 10_0,01_1,{\color{red}11_2},10_3,10_4,10_5\}
        \end{align} 
\end{itemize}

 Indeed, we check that 
    \[\cI_1\cup\cI_2 = \{\bfx_{i}: \bfx\in\cL',i\in[0,5]  \setminus \{2\} \}\cup\{\bfx_{2}: \bfx\in{\color{red}\FF_2^2}\}\,,\]

    % which corroborates with the claims in the proof of Theorem~\ref{thm:zigzag1}.
    Since we download the information symbols $\{\bfx_{i}: \bfx\in\cU\}$  from $\bfa^{(i)}$ for $i\in[0,5]  \setminus \{2\}$, respectively, we recovered all the values in $\{\bfx_{2}: \bfx\in{\color{red}\FF_2^2}\}$, i.e., $\bfa^{(2)}$.
    Since we download content only from the rows in $\cL'$, the skip cost is zero.

 Similarly, if the information node $\bfa^{(1)}$ is erased. From Fig.~\ref{fig:zigzag3}, the helper information nodes are $\bfa^{(i)}$ for $i\in [0,5]\setminus \{1\} $ ,
 while the helper parity nodes are also $\bfp^{(0)}$ and $\bfp^{(1)}$. 
  \begin{itemize}
    \item From $\boldsymbol{p}^{(0)}$, we obtain
\begin{align} \nonumber
\mathcal{I}_1& \triangleq 
\{00_0,{\color{red}00_1},00_2,00_3,00_4,00_5\} \cup  \{01_{0},{\color{red}01_1},01_2,01_3,01_4,01_5\} 
        \end{align}

\item From $\boldsymbol p^{(1)}$, we obtain
 \begin{align} \nonumber
 \mathcal {I}_2 &\triangleq  \{ 00_0,{\color{red}11_1},01_2,00_3,00_4,00_5\} \cup
   \{01_0, {\color{red}10_1}, 00_2,01_3,01_4,01_5\}
        \end{align} 
\end{itemize}
 
\noindent By downloading the information symbols
$\{\boldsymbol{x}_{i}: \boldsymbol{x}\in\mathcal{U}\}$ from $\bfa^{(i)}$ for $i\in [0,5]\setminus \{1\} $, it is clear that we can recover all values in $\{\boldsymbol{x}_{1}: \boldsymbol{x}\in\mathbb{F}_2^2\}$. Furthermore, we can check the repair scheme of other information nodes according to Fig.~\ref{fig:zigzag3}, and their skip costs are zero also.
\qedsymbol

\end{example}

%$$------------------$$

\section{SQS-Based Fractional Repetition Codes with Zero Skip Cost}
%\section{Steiner Quadruple Systems with Zero Skip Cost}
\label{sec:fractional}

In this section, we study fractional repetition codes with repair schemes that incur zero skip cost.
As discussed in Section~\ref{sec:prelim-fr}, to construct the resulting array code, it suffices to specify the {\em order} of points in each block. 
Therefore, in what follows, instead of writing blocks as unordered sets, we write each block explicitly as a {\em tuple}.

Also, as previously discussed, we focus on fractional repetition codes resulting from Steiner quadruple systems. 
Here, we describe two constructions. 
In Section~\ref{sec:sqs-recursive}, we look at recursive constructions that yield two infinite classes of SQS with zero skip cost.
Later, in Section~\ref{sec:sqs-differences}, we use the method of differences to construct more SQS with other parameters. 
Here, our SQS satisfy a more stringent property where every point-pair appears in at least two different blocks.

\subsection{Recursive Constructions}\label{sec:sqs-recursive}

Here, our constructions follow closely the original proof by Hanani in 1960 \cite{hanani1960quadruple}. Specifically, in that paper, Hanani provided six recursive constructions that build a larger SQS from smaller ones. Here, we make modifications to two of these recursive constructions so that the resulting SQS has locality two and skip cost zero.

\vspace{1mm}

\noindent{\bf Construction D} (Doubling Construction).
Let $(V,\cB)$ be an $\sqs(v)$ with $V=[v]$. We construct a $\sqs(2v)$-design on the point set $[v]\times \{0,1\}$.
The collection of blocks comprise two groups:
\begin{comment}
    \item $\cB_1$ is the collection $\bigg\{ \Big((v_1,i_1), (v_2,i_2), (v_3,i_3), (v_4,i_4)\Big) : \{v_1,v_2,v_3,v_3\}~\in~\cB,~\sum_{j=1}^4 i_j=0\pmod{2} \bigg \}$. In other words, $\cB_1$ contains $8|\cB|$ blocks;
    \item $\cB_2$ is the collection $\bigg\{ \Big((v_1,0), (v_1,1), (v_2,0), (v_2,1)\Big) : \{v_1,v_2\} \text{ is a two-subset of } V \bigg\}$. In other words, $\cB_2$ contains $\binom{N}{2}$ blocks.
\end{comment}
\begin{itemize}
    \item $\cB_1$ is the collection $\{ ((v_1,i_1), (v_2,i_2), (v_3,i_3), (v_4,i_4)) : \{v_1,v_2,v_3,v_4\}~\in~\cB,~\sum_{j=1}^4 i_j=0\pmod{2}  \}$. In other words, $\cB_1$ contains $8|\cB|$ blocks.
    \item $\cB_2$ is the collection $\{ ((v_1,0), (v_1,1), (v_2,0), (v_2,1)) : \{v_1,v_2\} \text{ is a two-subset of } V \}$. In other words, $\cB_2$ contains $\binom{N}{2}$ blocks.
\end{itemize}

Note that for any 3 elements $\{ (v_1,i_1), (v_2,i_2), (v_3,i_3) \}$, if $v_1\neq v_2\neq v_3\neq v_4$ it is included in $\cB_1$, and otherwise it is included in  $\cB_2$. Hence, every subset of $\sqs(2v)$ containing $3$
element is contained in exactly one block.

\noindent{\em Repair of Block $\bfb\in \cB_1$}: 
Then $\bfb$ is of the form $((v_1,i_1), (v_2,i_2), (v_3,i_3), (v_4,i_4))$. Then we pick the following reads and blocks.
 \begin{itemize}\itemindent=-10pt%\small
    \item ${\sf R}_1 = \{(v_1,i_1), (v_2,i_2)\}$ in block $((v_1,i_1), (v_2,i_2), (v_3,1-i_3), (v_4,1-i_4))$;
    \item ${\sf R}_2 = \{(v_3,i_3), (v_4,i_4)\}$ in block $((v_1,1-i_1), (v_2,1-i_2), (v_3,i_3), (v_4,i_4))$.
\end{itemize}

\noindent{\em Repair of Block $\bfb\in \cB_2$}:
Then $\bfb$ is of the form $((v_1,0), (v_1,1), (v_2,0), (v_2,1))$.
Let $v_3\notin\{v_1,v_2\}$ and we pick the following reads and blocks.
 \begin{itemize}\itemindent=-10pt\small
 \item ${\sf R}_1 = \{(v_1,0), (v_1,1)\}$ in block $((v_1,0), (v_1,1), (v_3,0), (v_3,1))$;
 \item ${\sf R}_2 = \{(v_2,0), (v_2,1)\}$ in block $((v_3,0), (v_3,1), (v_2,0), (v_2,1))$.
 \end{itemize} 

\begin{theorem}\label{thm:sqs-2v}
The repair scheme for Construction~D is correct. 
That is, for $\bfb\in\cB$, we are able to determine $\bfb$ from all reads.
Moreover, the repair scheme has locality two and skip cost zero.
\end{theorem}

\begin{proof}
    In all cases, we simply check that $\Read_1\cup\Read_2=\bfb$.
    The locality and skip cost properties are readily verified.
\end{proof}

\begin{example}
    Let $v=4$ with $V=[4]$. Clearly, if $\cB$ has the single block $(1,2,3,4)$, we have an $\sqs(4)$.
    Construction~D builds an $\sqs(8)$. 
    %Then, in $\cB_1$, we have the eight blocks
    %\begin{eqnarray*}{cccc}
    %    (1_0,2_0,3_0,4_0), & (1_0,2_0,3_1,4_1), & (1_0,2_1,3_0,4_1), & (1_0,2_1,3_1,4_0), \\
    %    (1_1,2_0,3_0,4_1), & (1_1,2_0,3_1,4_0), & (1_1,2_1,3_0,4_0), & (1_1,2_1,3_1,4_1). \\
    %\end{eqnarray*}
    Then we have eight and six blocks in $\cB_1$ and $\cB_2$, respectively.
    We present the placement of the 14 blocks in Fig.~\ref{fig:sqs}(b). 
    
    We illustrate the repair scheme assuming that the node / block $(1_0,1_1,2_0,2_1)$ has failed.
    Since the block belongs to $\cB_2$, we pick $v_3=3$ and read
\begin{itemize}%\itemindent=-10pt
    \item ${\sf R}_1 = \{1_0, 1_1\}$ from the block $(1_0,1_1,3_0,3_1)$;
    \item ${\sf R}_2 = \{2_0, 2_1\}$ from the block $(2_0,2_1,3_0,3_1)$.
    \end{itemize} 
    For both blocks, we see the repair has zero skip cost.
    
    In Fig.~\ref{fig:sqs}(b), we illustrate the repair for a block from $\cB_1$ and we see that the skip cost is zero too. \qedsymbol
\end{example}

Similarly, we modify another Hanani's recursive construction.
Unfortunately, the repair scheme becomes significantly involved as we have more cases to handle. Also, due to space constraints, we only describe the ordering of blocks and defer the details of the repair scheme to Appendix~\ref{app:repair-c}.

\vspace{1mm}
\noindent{\bf Construction E} ($(3v-2)$-Construction).
Let $(V,\cB)$ be an $\sqs(v)$ with $V=\{\infty\}\cup [N]$. We construct a $\sqs(3v-2)$-design on the point set $\{\infty\}\cup [N]\times \{0,1,2\}$.

The collection of blocks comprises six groups:
\begin{itemize}
\item $\cB_1$ is the collection $\{ ((i_1,v_1), (i_2,v_2), (i_3,v_3), (i_4,v_4)) : \{v_1,v_2,v_3,v_4\} \in \cB, ~\sum_{j=1}^4 i_j=0\pmod{3} \}$. In other words, $\cB_1$ contains $27|\cB|$ blocks.

\item $\cB^{1}_2$ is the collection $\{ ((i_1,v_1), \infty, (i_3,v_3), (i_2,v_2) ) : \{\infty,v_1,v_2,v_3\} \in \cB,~v_1< v_2< v_3,~\sum_{j=1}^3 i_j=0\pmod{3}, i_1= i_2 \}$ and
\item $\cB^{2}_2 = \{ ((i_1.v_1), \infty, (i_2,v_2), (i_3.v_3)) : \{\infty,v_1,v_2,v_3\} \in \cB, ~v_1< v_2< v_3, \sum_{j=1}^3 i_j=0\pmod{3} \}$. In other words, $\cB^{1}_1\cup \cB^2_2$ contains $9|\cB|$ blocks.

\item $\cB_{3}$ is the collection $\{ ((i,v_1), (i+1,v_2), (i,v_3), (i+2,v_2)) : \{\infty,v_1,v_2,v_3\} \in \cB, i\in \{0,1,2\} \}$. In other words, $\cB_3$ contains $9|\cB|$ blocks.

\item  $\cB_{4}$ is the collection $\{ ((i,v_1), (i+1,v_1), (i+1,v_2), (i,v_2)) : \{v_1,v_2\} \text{ is a two-subset of } [0,N-1] , i\in \{0,1,2\}  \}$, $v_1< v_2$.
In other words, $\cB_4$ contains $3\binom{N}{2}$ blocks.

\item $\cB_{5}$ is the collection $\{ \infty, ((0,v), (1,v), (2,v)) : v\in [N] \}$. In other words, $\cB_4$ contains $N$ blocks.
\end{itemize}

\begin{theorem}\label{thm:sqs-3v}
The repair scheme for Construction~E is correct. 
That is, for $\bfb\in\cB$, we are able to determine $\bfb$ from all reads.
Moreover, the repair scheme has locality two and skip cost zero.
\end{theorem}

\subsection{Method of Differences}\label{sec:sqs-differences}

Besides the quadruple systems described by Hanani, many other classes of SQS were studied (see survey by Lindner and Rosa~\cite{lindner1978steiner}). Of interest, we look at SQS with large automorphism groups.
For such systems of order $v$, the vertex set is defined over the cyclic group $\ZZ_v$ or $\ZZ_{v-1}$ with a fixed point $\infty$. In other words, $V=\ZZ_v$ or $V=\ZZ_{v-1}\cup\{\infty\}$.
Then we have a set of {\em base} blocks $\cS$ and all blocks are obtained via cyclic shifts of blocks in $\cS$.
In other words, $\cB = \{\bfb + i \,:\, \bfb\in\cS, i\in\ZZ_v\}$ or $\cB = \{\bfb + i \,:\, \bfb\in\cS, i\in\ZZ_{v-1}\}$.

Now, in this section, we provide quadruple systems with a stronger property.
\begin{definition}
An $\sqs(v)$ $(V,\cB)$ has {\em repeated adjacent pairs} if
for all $2$-subsets $P$ of $V$, the pair $P$ appears as an adjacent pair in at least two different blocks in $\cB$.
\end{definition}

Observe that if an $\sqs(v)$ has repeated adjacent pairs, then the $\sqs(v)$ has locality two and skip cost zero.
Next, we employ the {\em method of differences} to construct such quadruple systems.

\begin{definition}
Let $(a,b,c,d)$ be an ordered block defined over a cyclic group $\ZZ_g$. 
Its {\em difference list} $\diff(a,b,c,d)$ is given by $ \langle |b-a|, |c-b|, |d-c|\rangle$.
Here, $|b-a|= \min\{b-a\pmod g, a-b\pmod g\}$ and we use $\langle\cdot\rangle$ to denote a multi-set.

Given a collection of blocks $\cS$, we use $\diff(\cS)$ to denote $\bigcup_{\bfb\in \cS} \diff(\bfb)$\,.
As with difference lists, we are taking a multi-union of multi-sets.
\end{definition}

\begin{proposition}
\label{prop:sqs-differences}
Consider a $\sqs(v)$ with base blocks $\cS$.
\begin{enumerate}[(i)]
\item Suppose that $V=\ZZ_v$. If the difference $i$ appears at least twice in $\diff(\cS)$ for all $0<i<v/2$ and the difference $i$ appears at least once in  $\diff(\cS)$, then the $\sqs(v)$ has repeated adjacent pairs.
\item Suppose that $V=\ZZ_{v-1}\cup \{\infty\}$. If the difference $i$ appears at least twice in $\diff(\cS)$ for all $0<i\le (v-1)/2$ and $\infty$ appears at least two different blocks in $\cS$, then the $\sqs(v)$ has repeated adjacent pairs.
\end{enumerate}
\end{proposition}

Due to space constraints, we defer the proof to Appendix~\ref{app:sqs-differences}.
Instead, we provide an example for the case $v=26$.

\begin{example}\label{exa:sqs26}
We consider an $\sqs(26)$ defined over $\ZZ_{25}\cup\{\infty\}$.
Specifically, \cite{ji2002improved} provided the following 26 base blocks. 
Here, we reorder the points in some blocks in order to fulfill the conditions of Proposition~\ref{prop:sqs-differences}.
\[
\small
\arraycolsep=1.2pt
\begin{array}{llllll}
(0, 1, 3,\infty), &  (0, 4, 11, \infty), & (0, 5, 13, \infty), & (0, 6, 15, \infty), & (0, 1, 2, 5), \\
(0, 1, 6, 7), & (0, 1, 8, 9), &  (0, 1, 10, 11), & (0, 1, 12, 22), & (0, 1, 13, 21),\\ 
(0, 1, 14, 23), & (0, 2, 4, 12),& (0, 2, 6, 9), &  (0, 2, 7, 17), & (0, 2, 8, 22), \\
(0, 2, 11, 18), & (0, 2, 13, 19) & (0, 2, 14, 21), & (0, 2, 15, 20), &  (0, 3, 6, 10), \\
(0, 3, 8, 17), & (0, 3, 9, 14), & (0, 3, 12, 18), & (0, 3, 13, 20), & (0, 14, 8, 4), \\
(0, 4, 9, 13).
\end{array}
\]
For example, we observe that $\infty$ is in both the first two blocks.
The difference one appears twice in the block $(0,1,2,5)$,
while the difference five appears once in the block $(0,2,8,17)$ and once in the block $(0,4,9,13)$.
%Specifically, the blocks highlighted in {\color{blue}blue} contains all nonzero differences at least twice.
%\hm{Wenqin: can you add the blocks for $\sqs(26)$ here?}
\end{example}

\subsection{Proof for Theorem~\ref{thm:sqs-all}}

Recall that \cite{hanani1960quadruple} yields that an $\sqs(v)$ whenever $v\equiv 2$ or $4 \pmod{6}$ (see Theorem~\ref{thm:hanani}).
Applying Construction~B, we then obtain an $\sqs(V)$ with locality two and skip cost zero whenever $V\equiv 4$ or $8\pmod{12}$. Similarly, Construction~C yields an $\sqs(V)$ with locality two and skip cost zero whenever $V\equiv 4$ or $10\pmod{18}$.
These two constructions yield the infinite set of parameters in Theorem~\ref{thm:sqs-all}.

Next, for $v\in\{26,34,38\}$, we use Proposition~\ref{prop:sqs-differences}. Example~\ref{exa:sqs26} provides the base blocks for $\sqs(26)$, while Appendix~\ref{app:sqs-differences} provides the base blocks for $\sqs(34)$ and $\sqs(38)$.

Finally, we provide an explicit construction of a $\sqs(14)$ with zero skip cost in Appendix~\ref{app:sqs-14}. This completes the proof for Theorem~\ref{thm:sqs-all}.

\bibliographystyle{IEEEtran}
\bibliography{references}

\newpage

\phantom{SPACE TO FORCE NEW PAGE}

\appendices

\section{Proof of MDS Property of Construction A}
\label{app:mds}

Let us recall Construction A.
For $m\geq2$, set $k=m+1$, $M=2^m$ and $N=2k=2(m+1)$.
Then any codeword has $k=m+1$ information columns $\bfa^{(0)},\bfa^{(1)}\ldots, \bfa^{(m)}$, while the remaining $k=m+1$ are parity nodes $\bfp^{(0)}, \bfp^{(1)},\ldots, \bfp^{(m)}$.

Here, we use $\bfA$ to represent the $2^m\times k$ array $\left(a^{(i)}_{\bf x}\right)$.
Also, we use a column vector $\bfB$ of length $Mk$, where the column $i\in[0,k-1]$ of $\bfA$ is in the row set $ [iM,(i+1)M-1]$ of $\bfB$. 
Each systematic node $i$, or $\bfa^{(i)}$, can be represented as $\bfV_{i}\bfB$,  where $\bfV_i=[\vzero_{p\times pi}, \bfI_{p\times p},\ldots,\vzero_{p\times p(k-i-1)}]$.

For a fixed $i,j$ and any $\bfx\in \FF^m_2$, we set $\alpha^{(i,j)}_\bfx=x_{j,i}$. Similarly, each parity node $p^{(i)}$ can be obtained by $\bfW_i\bfB$, where $\bfW_{i}=[x_{i,0}\bfP_{i,0}, x_{i,1}\bfP_{i,1}, \ldots, x_{i,k-1}\bfP_{i,k-1}]$, $\bfP_{0,j}=\bfI_{p\times p}$ and $\bfP_{i\neq0,j}$ are permutation matrices (not necessarily distinct) of size $p\times p$ for $i,j\in[0,k-1]$.  
The permutation matrix $\bfP_{i,j}=(p^{(i,j)}_{{\bfx},{\bfy}})$ is defined as $p^{(i,j)}_{{\bfx},{\bfy}}=1$ if and only if $a^{(j)}_{\bfy}=a^{(0)}_{\bfx} +v_{i,j}$, where $v_{i,j}$ is the $(j+1)$-th element of set $\boldsymbol{S}_i$.

Now, to demonstrate the MDS property,
it suffices to show that there is an assignment for the indeterminates $\{x_{i,j}\}$ in the field $\FF_q$, such that for any subset $\{s_1,s_2,\ldots,s_k\} \subset [0,k-1]$, the matrix $\bfW=[\bfW^T_{s_1}, \bfW^T_{s_2},\ldots,\bfW^T_{s_k} ]$ is of full rank.

It is clear that $\bfW=[\bfV^T_{0},\cdots,\bfV^T_{k-1}]$ is of full rank. 
Let $t\in[1,k]$ be an integer.
Assume that $[\bfW^T_{i_1},\bfW^T_{i_2},\ldots,\bfW^T_{i_t}]$ is a submatrix of $\bfW$, i.e.,  $\{i_1,i_2,\ldots,i_t\} \subset 
\{s_1,s_2,\ldots,s_k\} $. This means that there exists $0\leq j_1<j_2<\ldots<j_t \notin\{s_1,s_2,\ldots,s_k\}$ such that $V^T_{j_1}, V^T_{j_2},\ldots V^T_{j_t}$ is not a submatrix of $W$. 
 In this case, $\bfW$ is of full rank if and only if the submatrix
\begin{equation}\label{PCmatrix}
    \bfD_{T^1,T^2}=\left(\begin{array}{ccccccccc}
 x_{i_1j_1}\bfP_{i_1j_1}  &   x_{i_1j_2}\bfP_{i_1j_2} &\ldots &   x_{i_1j_t}\bfP_{i_1j_t}  \\
  x_{i_2j_1}\bfP_{i_2j_1} & x_{i_2j_2}\bfP_{i_2j_2} & \cdots &   x_{i_2j_t}\bfP_{i_2j_t} \\
   \vdots  &  \vdots  & \ddots &   \vdots \\
x_{i_tj_1}\bfP_{i_tj_1}  & x_{i_tj_1}\bfP_{i_tj_1} & \ldots &   x_{i_tj_t}\bfP_{i_tj_t}
    \end{array}\right),
\end{equation}
is of full rank, where $T^1=\{i_1,i_2,\ldots,i_t\}$, $T^2=\{j_1,j_2,\ldots,j_t\}$. This is equivalent to $\det(\bfD_{T^1,T^2})\neq 0$. Note that $\deg(\det(\bfB_{T^1,T^2}))=pt$ and the coefficient of $\prod^{t}_{z=1}x_{i_zj_z}^p$ is $\det(\prod^{t}_{z=1}\bfP_{i_zj_z})\in\{1,-1\}$.

Let $\mathcal{T}=\{\{i_1,i_2,\ldots,i_t\} \mid 0\leq i_1<i_2<\ldots<i_t\leq k-1 \}$. 
 Obviously, $|\mathcal{T}|=\binom{k}{t}$.
Let $T^1$ and $T^2$ be a subset of set $\mathcal{T}$ with size $t$, and 
define the polynomial 
\begin{equation}
Y(x_{0,0},x_{0,1},\ldots,x_{k-1,k-1})=\prod_{T^1\in \mathcal{T}}  \prod_{T^2\in \mathcal{T}} \det (\bfD_{T^1,T^2})\,.
\end{equation}
The result follows if there are elements $\lambda_{0,0},\ldots,\lambda_{k-1,k-1}$ in $\FF$
such that $Y(x_{0,0},x_{0,1},\ldots,x_{k-1,k-1})\neq~0$.
 Assume that there exists a largest nonzero term $\prod^{k-1}_{i=0,j=0} x_{i,j}^{y_{i,j}}$ in the polynomial $Y(x_{0,0},x_{0,1},\ldots,x_{k-1,k-1})$, where $\deg(x_{i,j})=y_{i,j}$. 
 We apply the Combinatorial Nullstellensatz~\cite{alon1999combinatorial}, If $q\geq \max\{y_{i,j}\}+1$, then there exists elements $\lambda_{0,0}$,$\ldots$,
$\lambda_{k-1,k-1}\in\FF$
such that $Y(x_{0,0},x_{0,1},\ldots,x_{k-1,k-1})\neq 0$.

First, we consider the selection of columns of the given parity matrices.
For a given subset $T^1=\{i_1,i_2,\ldots,i_t\}$ and any subset $T^2$, we have
$\deg(x_{i_1,0})=\binom{k-1}{t-1}$, $\deg(x_{i_1,1})=\binom{k-2}{t-1}$, $\cdots$, $\deg(x_{i_1,k-t})=\binom{t-1}{t-1}$,
$\deg(x_{i_2,1})=\binom{k-2}{t-2}$, $\deg(x_{i_2,2})=\binom{2}{1}\binom{k-3}{t-2}$, $\deg(x_{i_1,3})=\binom{3}{1}\binom{k-4}{t-2}$, $\cdots$, $\deg(x_{i_1,k-t+1})=\binom{t-2}{t-2}\binom{k-t+1}{1}$,
$\deg(x_{i_3,2})=\binom{k-3}{t-3}$, $\deg(x_{i_3,3})=\binom{3}{2}\binom{k-4}{t-3}$, $\deg(x_{i_3,4})=\binom{4}{2}\binom{k-5}{t-3}$, $\cdots$, $\deg(x_{i_3,k-t+2})=\binom{t-3}{t-3}\binom{k-t+2}{2}$,
$\cdots$,
$\deg(x_{i_t,t-1})=\binom{k-t}{0}$, $\deg(x_{i_t,t})=\binom{t}{t-1}\binom{k-t-1}{0}$, $\deg(x_{i_t,t+1})=\binom{t+1}{t-1}\binom{k-t-2}{0}$, $\cdots$, $\deg(x_{i_t,k-1})=\binom{k-1}{t-1}\binom{0}{0}$.
Note that for the $i_\ell$ parity matrix,  when $ j \leq \ell-2$, the term $x_{i_\ell,j}$ does not exist; when $\ell-1 \leq j \leq k-t+\ell-1$, the term $x_{i_\ell,j}$ exists. Therefore, there exists a term $\prod^t_{\ell=1}\prod^{k-t+\ell-1}_{j=\ell-1}x_{i_{\ell},j}^{\binom{j}{\ell-1}\binom{k-(j+1)}{t-\ell}p}$ with the largest degree being $\det(\bfD_{T^1,T^2})$.

Next, we consider the case of any subset $T^1$.
For a given $0\leq \beta \leq k-1$, we need to calculate the number $Q_{\beta,\ell}$ of the subset $T^1=\{i_1,i_2,\cdots,i_t\}$ that satisfies $i_\ell=\beta$ for $\ell\in[t]$. According to the sorting order, $\beta$ appears in a position whose index is less than $\beta+1$ of subset $T^1$, i.e., $\ell\leq\beta+1$. Hence, we obtain $
    Q_{\beta,\ell}=\binom{\beta}{\ell-1} \binom{k-(\beta+1)}{t-\ell}$.
For example, if $\beta=0$, $\beta$ occurs only in the first position of each subset $T^1$, i.e., $i_1=0$. Therefore, we have $Q_{\beta=0,\ell=1}=\binom{k-1}{t-1}$. It is clear that $\deg(x_{0,j})=p\binom{k-1}{t-1} \binom{k-(j+1)}{t-1}$ for $j\in[0,k-1]$.

For fixed $\beta\in[0,k-1]$ and $j\in[0,k-1]$, we find that the largest degree of $x_{\beta,j}$ in a term $\prod^{k-1}_{i,j=0}x_{i,j}^{y_{i,j}}$  is
 $\sum^{t}_{\ell=1}p Q_{\beta,\ell} \binom{j}{\ell-1}\binom{k-(j+1)}{t-\ell}=\sum^{t}_{\ell=1} p\binom{\beta}{\ell-1} \binom{k-(\beta+1)}{t-\ell} \binom{j}{\ell-1}\binom{k-(j+1)}{t-\ell} $ .
By the following inequality
{\footnotesize\begin{eqnarray}\label{eq:com}
\nonumber
& & \sum^{t}_{\ell=1} \binom{\beta}{\ell-1} \binom{k-(\beta+1)}{t-\ell} \binom{j}{\ell-1}\binom{k-(j+1)}{t-\ell} \\ \nonumber
 &\leq&   \left(\sum^{t}_{\ell=1} \binom{\beta}{\ell-1} \binom{k-(\beta+1)}{t-\ell} \right)\left(\sum^{t}_{\ell=1} \binom{j}{\ell-1}\binom{k-(j+1)}{t-\ell}  \right) \\  \nonumber
&\leq & \binom{k-(\beta+1)+\beta}{t-1}\binom{k-(j+1)+j}{t-1}=\binom{k-1}{t-1}^2,
\end{eqnarray}}
\noindent\hspace*{-3mm} we have that $\max\{y_{i,j}\}= p\binom{k-1}{t-1}^2$.
The equation holds due to the identity $\sum^k_{i=0}\binom{n}{i}\binom{m}{k-i}=\binom{n+m}{k}$.  In fact, we can find a largest term $x_{0,0}^{ p\binom{k-1}{t-1}^2}\prod^{k-1}_{i,j=1}y_{i,j}$ of the polynomial $Y(\lambda_{0,0},\lambda_{0,1},\ldots,\lambda_{k-1,k-1})$ and its coefficient is nonzero as the result of the product of the permutation matrix.
When $t=\lfloor\frac{k-1}{2}\rfloor+1=\lceil\frac{k}{2}\rceil$, we have the maximum value $p\binom{k-1}{t-1}^2$ such that the function 
$Y(\lambda_{0,0},\lambda_{0,1},\ldots,\lambda_{k-1,k-1})\neq 0 $.

\section{Repair Scheme for Construction C}
\label{app:repair-c}

\noindent{\em Repair of Block $\bfb\in \cB_1$}: 
Then $\bfb$ is of the form $((v_1,i_1), (v_2,i_2), (v_3,i_3), (v_4,i_4))$. Then we pick the following reads and blocks.
\begin{itemize}
    \item ${\sf R}_1 = \{(i_1,v_1), (i_2,v_2)\}$ in the block $\bfb_1=((i_1,v_1), (i_2,v_2), (i_3+1,v_3), (i_4+2,v_4))$;
    \item ${\sf R}_2 = \{(i_3,v_3), (i_4,v_4)\}$ in the block $\bfb_2=((i_1+1,v_1), (i_2+2,v_2), (i_3,v_3), (i_4,v_4))$.
\end{itemize}

\noindent{\em Repair of Block $\bfb\in \cB^1_2$}: 
Then $\bfb$ is of the form $((i_1,v_1), \infty, (i_3.v_3), (i_2,v_2) )$ with $i_1=i_2$. Then $i_1+i_2+i_3=0\pmod{3}$ implies that $i_1=i_2=i_3$.
\begin{itemize}
    \item ${\sf R}_1 = \{ (i_1,v_1),\infty\}$ in the block $\bfb_1=((i_1,v_1), \infty, (i_2+1,v_2), (i_3+2,v_3)) \in \cB_2^2$;
    \item ${\sf R}_2 = \{(i_2,v_3), (i_2,v_2)\}$ in the block $\bfb_2=(i_2-1,v_2), (i_2,v_2), (i_2,v_3), (i_2-1,v_3))\in \cB_4$.
\end{itemize}

\noindent{\em Repair of Block $\bfb\in \cB^2_2$}: 
Then $\bfb$ is of the form $((i_1,v_1), \infty, (i_2,v_2), (i_3,v_3))$ with $i_1\ne i_2$. Then $i_1+i_2+i_3=0\pmod{3}$ implies that $\{i_1,i_2,i_3\}=\{0,1,2\}$.
\begin{itemize}
    \item ${\sf R}_1 = \{(i_1,v_1), \infty \}$ in the block $\bfb_1=((i_1,v_1), \infty, (i_1,v_2), (i_1,v_3)) \in \cB_2^1$;
    \item If $i_3=i_2+1$, then ${\sf R}_2 = \{(i_3,v_3), (i_2,v_2)\}$ in the block $\bfb_2=((i_2+1,v_1), (i_2+2,v_2), (i_2+1,v_3), (i_2,v_2))\in \cB_3$.
    \item If $i_3=i_2+2$, then ${\sf R}_2 = \{(i_3,v_3), (i_2,v_2)\}$ in the block $\bfb_2=((i_2+2,v_1), (i_2,v_2), (i_2+2,v_3), (i_2+1,v_2))\in \cB_3$.
\end{itemize}

\noindent{\em Repair of Block $\bfb\in \cB_3$}: 
Then $\bfb$ is of the form $((i,v_1), (i+1,v_2), (i,v_3), (i+2,v_2))$. Then we pick the following reads and blocks.
\begin{itemize}
    \item ${\sf R}_1 = \{(i,v_1), (i+1,v_2)\}$ in the block $\bfb_1=((i+1,v_3), (i+2,v_1), (i+1,v_2), (i,v_1))\in \cB_3$;
    \item ${\sf R}_2 = \{(i,v_3), (i+2,v_2)\}$ in the block $\bfb_2=((i+2,v_2), (i,v_3), (i+2,v_1), (i+1,v_3))\in \cB_3$.
\end{itemize}

\noindent{\em Repair of Block $\bfb\in \cB_4$}: 
Then $\bfb$ is of the form $\bfb=((i,v_1), (i+1,v_1), (i+1,v_2), (i,v_2))$. Then we have the following three sub-cases.

{\em If $v_1=0, v_2= 1$}, then $\bfb=((i,0), (i+1,0), (i+1,1), (i,1))$. Let $v\in [0,N]$.  Then we pick the following reads and blocks.
\begin{itemize}
    \item ${\sf R}_1 = \{(i,0), (i+1,0)\}$ in the block $\bfb_1=((i,0), (i+1,0), (i+1,v), (i,v))\in \cB_4$, where $v > 1$.
    \item ${\sf R}_2 = \{(i+1,1), (i,1)\}$ in the block $\bfb_1=((i,1), (i+1,1), (i+1,v), (i,v)\in \cB_4$, where $ v > 1$.
\end{itemize}

{\em If $v_1= N-2$}, then $v_2=N-1$ and $\bfb=((i,N-2), (i+1,N-2), (i+1,N-1), (i,N-1))$.
\begin{itemize}
    \item ${\sf R}_1 = \{(i,N-2), (i+1,N-2)\}$ in the block $\bfb_1=((i,v), (i+1,v), (i+1,N-2), (i,N-2))\in \cB_4$, where $v < N-2$.
    \item ${\sf R}_2 = \{(i+1,N-1), (i,N-1)\}$ in the block $\bfb_1=((i,v), (i+1,v), (i+1,N-1), (i,N-1)\in \cB_4$, where $v < N-1$.
\end{itemize}

{\em If $(v_1,v_2)\ne (0,1)$ and $v_1\ne N-2$}, then $\bfb=((i,v_1), (i+1,v_1), (i+1,v_2), (i,v_2))$.
\begin{itemize}
    \item ${\sf R}_1 = \{(i,v_1), (i+1,v_1)\}$ in the block $\bfb_1=((i,v_1), (i+1,v_1), (i+1,v), (i,v))\in \cB_4$, where $v \ne v_2$.
    \item ${\sf R}_2 = \{(i+1,v_2), (i,v_2)\}$ in the block $\bfb_1=((i,v), (i+1,v), (i+1,v_2), (i,v_2))\in \cB_4$, where $v \ne v_1$.
\end{itemize}

\noindent{\em Repair of Block $\bfb\in \cB_5$}: 
Then $\bfb$ is of the form $( \infty, (0,v),(1,v), (2,v))$. Let ${\sf R}_1 = \{\infty,(0,v)\}$ and ${\sf R}_2=\{(1,v), (2,v)\}$.  Then we pick the following reads and blocks.
\begin{itemize}
    \item For the second read ${\sf R}_2=\{(1,v), (2,v)\}$ of  $\bfb$ can be obviously repaired from $\cB_4$. The specific repair process is as follows:
    \begin{itemize}
        \item If $v = 0$, then ${\sf R}_2 = \{(1,0), (2,0)\}$ in the block $\bfb_2=((1,0), (2,0), (2,v_1), (1,v_1))\in \cB_4$, where $v_1 > 0$.
        \item If $v < N-1$, then ${\sf R}_2 = \{(1,v), (2,v)\}$ in the block $\bfb_2=((1,v), (2,v), (2,v_1), (1,v_1))\in \cB_4$, where $v_1 > v$.
        \item If $v = N-1$, then ${\sf R}_2 = \{(1,N-1), (2,N-1)\}$ in the block $\bfb_2=((1,v_1), (2,v_1+1), (2,N-1), (1,N-1))\in\cB_4 $, where $v_1 < N-1$.
    \end{itemize}

    \item For the first read ${\sf R}_1$, for a given $v\in [0,N-1]$, there exists a subset of $v\in\{ \infty, v, v_1, v_2\}\in \cB$,
    \begin{itemize}
        \item If $v< v_1 < v_2$,  then ${\sf R}_1 = \{ \infty,(0,v)\}$ in the block $\bfb_1=((0,v), \infty, (0,v_2), (0,v_1))\in \cB_2^1$;
        \item If $v_1< v < v_2$,  then ${\sf R}_1 = \{ \infty, (0,v)\}$ in the block $\bfb_1=((1,v_1), \infty, (0,v), (2,v_2))\in \cB_2^2$;
        \item If $v_1< v_2 < v$,  then ${\sf R}_1 = \{ \infty, (0,v)\}$ in the block $\bfb_1=((0,v_1), \infty, (0,v), (2,v_2))\in \cB_2^1$.
    \end{itemize}
\end{itemize}

\section{SQS from Method of Differences}
\label{app:sqs-differences}

Here, we prove Proposition~\ref{prop:sqs-differences}(i) and then list down the base blocks for $\sqs(34)$ and $\sqs(38)$.

\begin{proof}[Proof of Proposition~\ref{prop:sqs-differences}(i)]
    We remark that the proof of Proposition~\ref{prop:sqs-differences}(ii) is similar and hence, omitted.
    Consider any pair $(x,y)$ and let $z = |x-y|$.
    From the conditions of the proposition, suppose that $z$ appears as a difference in blocks $\bfb_1$ and $\bfb_2$ (with possibly $\bfb_1=\bfb_2$).
    
    Suppose that $\bfb_1=(a,b,c,d)$ with $b=a+z$ and $y=x+z$.
    Then we choose $i=x-a$ and the block $\bfb+i=(x,b+x-a,c+x-a, d+x-a)=(x,y,c+x-a,d+x-a)$.
    Therefore, $(x,y)$ belongs to the block $\bfb_1+i$. Similarly, we can find another shift $j$ such that $\bfb_2+j$ contains $(x,y)$. In other words, $(x,y)$ belongs to different blocks and this concludes the proof. 
\end{proof}

Finally, we provide the base blocks for $\sqs(34)$ and $\sqs(38)$ that have repeated adjacent pairs.
As before, the base blocks given here are obtained from~\cite{ji2002improved}.

%\hm{Wenqin: can you just list the blocks here? Thanks a lot.}

\vspace{3mm}

\noindent{\em Base blocks for $\sqs(34)$}:
\[
\small
\arraycolsep=1.2pt
\begin{array}{llllll}
(0, 11, 22,\infty),^* &  (0, 1, 5, \infty), & (0, 2, 10, \infty), & (0, 3, 15, \infty), & (0, 6, 19, \infty ), \\
(0, 7, 16, \infty), & (0, 1, 2, 4), &  (0, 1, 6, 7), & (0, 1, 8, 9), & (0, 1, 10, 11), \\
(0, 1, 12, 13), & (0, 1, 14, 15), & (0 ,1 ,16, 29), &  (0, 1, 17, 31), & (0, 1, 18, 30), \\
(0, 2, 5, 7), & (0, 2, 6, 8), & (0, 2 ,9, 11), & (0, 2 ,12 ,14), &  (0, 2, 13, 16), \\
(0, 2, 15, 17), & (0, 2, 22, 25), & (0, 3, 6, 26), & (0, 3, 7, 25), & (0, 3, 8, 28), \\
(0, 3, 9, 17),&(0, 18, 10, 3),  & (0, 3, 12, 27),& (0, 3, 14, 29),& (0, 3, 16, 24), \\  
(0, 4, 8, 16),&(0, 4, 9, 28), & (0, 4, 10, 24),& (0, 4, 11, 25),& (0, 4, 13, 26),\\
(0, 19, 14, 4),&(0, 4, 15, 23),& (0, 4, 17, 27),& (0, 5, 10, 17), & (0, 5, 11, 21), \\
(0, 5, 15, 26), & (0, 16, 5, 22),&(0, 5, 18, 27), &(0, 6, 12, 21),& (0, 6, 13, 25),\\
(0, 7, 14, 24).
\end{array}\]
Here, the points are defined over $\ZZ_{33}\cup\{\infty\}$. The block marked with $*$ is a special base block that generates $33/3=11$ blocks.
All other base blocks generates $33$ blocks.
%As before, highlighted in \highlight{blue} are the blocks that contain all nonzero repeated differences.
%$\{0, v/3, 2v/3, \infty\}$  which generates $v/3$ blocks while all other base blocks each generates $v$ blocks.

\vspace{3mm}
\noindent{\em Base blocks for $\sqs(38)$}:
% under the automorphism group
% $ \{x \rightarrow mx+ b : m \in M=\{1,10,26\}, b \in \ZZ_{37}\}$
\[
\footnotesize
\arraycolsep=1.2pt
\begin{array}{llllll}
(0, 1, 27,\infty), &  (0, 2, 22, \infty), & (0, 3, 33, \infty), & (0, 5, 24, \infty), & (0, 6, 29, \infty), \\
(0, 9, 25, \infty), & (0, 1, 4, 11), & (0, 2, 17, 31), & ( 0, 3, 7, 28),& (0, 5, 18, 20),\\
(0, 6, 14, 19),& (0, 9, 10, 21 ),&(0,1,2,6),& (0,1,3,14),&(0,1,7,19),\\
(0, 1 ,9 ,35),& (0, 1, 10, 24),&(0, 30, 13, 1),&(0, 1, 15, 16),& (0, 1, 17, 33),\\
(0, 1, 20, 34),& (0, 1, 25, 29), &(0, 2, 4, 16),&(0, 2, 5, 8),&(0, 2, 7, 21),\\ 
(0, 2, 15, 18),& (0, 2 ,25, 30),&(0, 10, 20, 23),& (0, 10, 30, 29),& (0, 10, 33, 5),\\
(0, 10, 16, 17),& (0, 10, 26, 18),& (0, 10, 19, 4),& (0, 10, 2, 12),& (0, 10, 22, 34),\\
(0, 10, 15, 7),&(0, 10, 28, 31),& (0, 20, 3, 12),& (0, 20, 13, 6),& (0, 20, 33, 25), \\
(0, 20, 2, 32),& (0, 20, 28, 4),&(0, 26, 15, 8),& (4, 0, 26, 31),& (0, 26, 34, 13),\\
(0, 26, 12, 22),& (1, 0, 26, 32),& (0, 26, 5, 3),&(0, 26, 20, 9),& (0, 26, 35, 7),\\
(2, 0, 26, 33),& (0, 26, 21, 14),& (0, 15, 30, 9),& (0, 15, 19, 23),& (0, 15, 34, 28),\\
(0, 15, 20, 24),& (0, 15, 21, 3).
\end{array}
\]
Here, the points are defined over $\ZZ_{37}\cup\{\infty\}$. 
%As before, highlighted in \highlight{blue} are the blocks that contain all nonzero repeated differences.

\section{$\sqs(14)$ with Locality Two and Zero Skip Cost}
\label{app:sqs-14}

Here, we simply list down the blocks for $\sqs(14)$ with 14 symbols $\{0,1,\ldots,9,A,B,C,D \}$.
We remark that these blocks were first given in Hanani~\cite{hanani1960quadruple} and we adjusted the arrangement to obtain the desired skip cost.
\[
\footnotesize
\arraycolsep=1pt
\begin{array}{llllll}
(0,1,2,5), & (0,3,8,D), & (1,2,3,6), & (1,5,7,C), & (2,4,A,C), & (3,5,7,9),\\
(4,7,9,C), &(0,1,3,B), & (0,3,9,A), & (1,2,4,7), & (1,5,8,9), & (2,4,B,D), \\
(3,5,8,B), &  (5,6,7,8), & 
(0,1,4,6), &  (0,4,5,9), &  (1,2,8,B), &  (1,5,B,D),\\
(2,5,7,B), &  (3,5,A,C),&  (5,6,9,B), & 
(0,8,1,7), &  (0,4,7,B), &  (1,2,9,A), \\
(1,6,7,9), &  (2,5,8,A), &  (3,6,7,B), & (5,D,6,C), & 
(0,1,9,D), &  (0,4,8,A),\\
(1,2, C,D), &  (1,6,8,D), &  (2,5,9,C), &  (3,6,8,9), &  (5,9,A,D), &
(0,1,A, C), \\
(0,4,C,D), &   (1,3,4,5), &  (1,6,B,C), &  (2,6,7, C), &  (3,6,A, D), &  (6,8,A,C),\\
(0,2,3,4), &  (0,5,7,D), & (1,3,7,D), &  (1,7,A,B), &  (2,6,9,D), & (3,B,C,D), \\
(7,8,9,A), & (0,2,6,8), &  (0,5,8,C), &  (1,3,8,A), &  (2,3,5,D), &  (2,6,A,B), \\
 (4,5,7,A), &  (7,8,B,C), & 
(0,2,7,9), &  (0,5,A,B), &  (1,3,9,C), & (2,3,7,A),\\
(2,7,8,D), &  (4,5,8,D), &  (7,9,B,D),& 
(0,2,A,D), &  (0,6,7,A), &  (1,4,8,C),\\
(2,3,8,C), &  (3,4,6,C), &  (4,5,B,C),&  (7,A, C,D), & 
(0,2,B,C), &  (0,6,9,C),\\
(1,4,9,B), &  (2,3,9,B), &  (3,4,7,8), &  (4,6,7,D), &  (8,9,C,D), &  (0,3,5,6),\\
(0,6,B,D),&  (1,4,A,D), &  (2,4,5,6), &  (3,4,9,D), & (4,6,8,B), &  (8,A,B,D),\\
(0,3,7,C), &  (0,8,9,B),&  (1,5,6,A), &  (2,4,8,9), &  (3,4,A, B), &  (4,6,9,A), \\
(9,A,B,C).
\end{array}
\]
%\hm{Wenqin: can you just list the blocks here? Thanks a lot.}

\end{document}